\documentclass[10pt,journal,  oneside, twocolumn]{IEEEtran}

\usepackage{multirow}
\usepackage{latexsym}
\usepackage{graphicx}
\usepackage{float}
\usepackage{amsmath}
\usepackage{amsthm}
\usepackage{lipsum}
\usepackage{subfig}
\usepackage{graphicx}
\usepackage{authblk}
\usepackage{bm}
\usepackage{amsthm}

\usepackage{colortbl}

\newtheorem*{remark}{Remark}
\newtheorem{theorem}{Theorem}

\newtheorem{lemma}{Lemma}

\makeatletter  
\newif\if@restonecol  
\makeatother

\usepackage[linesnumbered,ruled,vlined]{algorithm2e}
\usepackage{algpseudocode}  
\usepackage{amsmath}

\usepackage{amssymb}

\usepackage{mathrsfs}
\usepackage{subfig}
\usepackage{caption}
\begin{document}

\title{Communication and Computing Resource Optimization for Connected Autonomous Driving}


\author{Kai Xiong, Supeng Leng,~\IEEEmembership{Member,~IEEE},  
Xiaosha Chen, Chongwen Huang, Chau Yuen,~\IEEEmembership{Senior Member,~IEEE} and Yong Liang Guan,~\IEEEmembership{Senior Member,~IEEE}

\thanks{

K. Xiong, S. Leng, and X. Chen are with School of Information and Communication Engineering, University of Electronic Science and Technology of China, Chengdu, 611731, China.}

\thanks{C. Huang and C. Yuen is with Singapore University of Technology and Design (SUTD), Singapore.}

\thanks{Y. L. Guan is with the School of Electrical and Electronic Engineering, Nanyang Technological University, Singapore.}

\thanks{The corresponding author is Supeng Leng, email: spleng@uestc.edu.cn}
}


\maketitle

\begin{abstract}

Transportation system is facing a sharp disruption since the Connected Autonomous Vehicles (CAVs) can free people from driving and provide good driving experience with the aid of Vehicle-to-Vehicle (V2V) communications. Although CAVs bring benefits in terms of driving safety, vehicle string stability, and road traffic throughput, most existing work aims at improving only one of these performance metrics. However, these metrics may be mutually competitive, as they share the same communication and computing resource in a road segment. From the perspective of joint optimizing driving safety, vehicle string stability, and road traffic throughput, there is a big research gap to be filled on the resource management for connected autonomous driving. In this paper, we first explore the joint optimization on driving safety, vehicle string stability, and road traffic throughput by leveraging on the consensus Alternating Directions Method of Multipliers algorithm (ADMM). However, the limited communication bandwidth and on-board processing capacity incur the resource competition in CAVs. We next analyze the multiple tasks competition in the contention based medium access to attain the upper bound delay of V2V-related  application offloading. An efficient sleeping multi-armed bandit tree-based algorithm is proposed to address the resource assignment problem. A series of simulation experiments are carried out to validate the performance of the proposed algorithms.

\end{abstract}


\IEEEpeerreviewmaketitle

\section{Introduction}

\IEEEPARstart {A}{utonomous} vehicles have been widely regarded as a promising technology to address great challenges in the intelligent transportation systems, such as driving safety, vehicle string stability (related to ride comfort), and road traffic throughput \cite{8320295}. However, the gains of autonomous driving are determined by the accuracy of the on-board sensors (e.g. Radar, camera, GPS) \cite{Kato8584062}. Nevertheless, these on-board sensors are usually costly as the equipment of an individual vehicle, if a vehicle does not equip complete on-board sensors, the inaccurate sensing information and traffic estimation may incur serious accidents. Vehicle-to-Vehicle (V2V) communications compensate for the perceived deficiency of an individual vehicle since the perception region of V2V communication is usually much larger than that of on-board sensors. The on-board processors and information can be shared for many vehicles in a large region. Thus, using shared computing and information, the Connected Autonomous Vehicle (CAV) not be necessarily equipped with complete on-board sensors for the cost reduction \cite{liulin2018}.

Driving safety and road traffic throughput are always the main goal for the driving strategy of CAVs. As the most important aspect of autonomous vehicle systems, driving safety is achieved by maintaining a suitable safety distance. The safety distance is defined as the inter-vehicle spacing, with which a crash can be avoided. Decreasing the inter-vehicle spacing is an effective way to increase road traffic throughput. In this case, platoons formed by vehicles with the same driving speed and direction have the potential to increase road traffic throughput by allowing small inter-vehicle spacing. However, the dense road traffic with the small inter-vehicle spacing can easily incur a back-and-forth velocity oscillation in the platoon, which deteriorates an important platoon metric, vehicle string stability. The string stability is commonly associated with the acceleration frequency and amplitude of vehicles in platoon \cite{Dunbar5876300}. Maintaining vehicle string stability can lead to a good driving safety. Therefore, the transportation management system should account for safe driving, road traffic throughput, and vehicle string stability, simultaneously.



It is known that accelerations of CAVs are controlled by the Cooperative Adaptive Cruise Control (CACC) system, which adaptively changes the vehicle velocity to maintain a suitable inter-vehicle spacing through Vehicle-to-Vehicle (V2V) communication \cite{7349170}. However, the inter-vehicle spacing among CAVs is associated with the V2V bandwidth (communication resource) and on-board processing capacity (computing resource). Sufficient communication and computing resources can shorten the safety distance through reducing the the data transmission delay and the on-board processing delay \cite{8798668}. Unfortunately, due to the random distribution of vehicles, the communication bandwidth in some road segments with low road traffic is under-utilized, while the bandwidth is exhausted in other road segments with road traffic jams. Additionally, the high mobility of vehicles deteriorates the intermittent V2V communication \cite{Qiao044}. Previous work on CACC systems ignored the unbalance and uncertainty of the resources allocated in CAVs, which may result in undesired inter-vehicle spacing to deteriorate road traffic.

This paper focuses on the exploration of the relation between safety distance, string stability, and road traffic throughput. Next, we propose an effective inter-vehicle spacing control scheme to optimize the string stability and road traffic throughput upon safe driving through resource management. The contributions of this paper can be summarized as follows:

\begin{itemize}

\item Based on the continuum road traffic model, we propose a mathematical model of the string stability in an account of safety distance. We further propose an optimal inter-vehicle spacing scheme that jointly optimizes road traffic throughput and string stability upon safe driving. 



\item By leveraging the theory of network calculus, we derive the closed-form of the upper bound of V2V offloading delay for the contention-based medium access control approaches such as IEEE 802.11p. This upper bound can provide a criterion to decide whether a vehicle has deficient communication and computing resource. 




\item According to the upper bound, we design an efficient communication and computing resource management approach that allocates the excess resource from the resource-rich vehicles to the resource-deficient vehicles in order to obtain the desired inter-vehicle spacing among all CAVs. This resource management approach is able to significantly reduce the execution time.


\end{itemize}

The remainder of this paper is organized as follows. Section II reviews the related work. Section III presents the continuum road traffic metrics. Section IV proposes the joint optimization scheme. Section V gives the upper bound of the V2V offloading delay. Section VI presents the resource allocation scheme. Section VII demonstrates simulation results and discussion. Finally, we draw the conclusion in Section VIII. A summary of the important mathematical notations used in the paper is given in Table \ref{tab_1}.

{\color{blue}

\begin{table}[h]
\caption{A summary of important mathematical notations.}
\begin{tabular}{p{1cm}|p{7cm}}
\hline
\hline
Symbol & Description \\
\hline
$s^*$ & The safety distance\\
\hline
$\tau_0$ & The perception-reaction delay\\
\hline
$T_{(ij)k}$ & The upper bound of the V2V offloading delay for the application $k$ from vehicle $i$ to vehicle $j$\\
\hline
$\rho$ & The vehicle density (number of vehicles in an unit area)\\
\hline
$\bm{f}$ & The flux of vehicle string\\
\hline
$\mathscr{F} $ & The road traffic throughput \\
\hline
$R_i$ & The available communication bandwidth of road segment $i$\\
\hline
$\delta$ & The weight of string stability in optimization model\\
\hline
\hline
\end{tabular} \label{tab_1}
\end{table}

}

\section{Related Work}

The advantage of the V2V based cooperative perception enables an individual vehicle to have a longer perception range, which can be used in the cooperative collision detection \cite{Gallego8845107} and lane changing warning. Nunen \textit{et al.} \cite{8317758} proposed an intended acceleration prediction based on V2V perception, which demands sufficient time to ensure high performance and robustness in terms of string stability and driving safety. However, the mathematical relation between string stability and V2V cooperative perception is not identified. Nekoui \textit{et al.} \cite{Nekoui2010Fundamental} demonstrated the V2V communication improves road traffic throughput by reducing driver Perception-Reaction Time (Delay), which implies a high-speed compact platoon. However, this work is not accounted for the impact of the string stability on the vehicle string. The high-speed compact platoons easily result in the vehicle string shockwaves that adversely affects the ride comfort \cite{Dunbar5876300}. Although these work discuss the V2V based perception impacting on road traffic metrics, they ignored to jointly optimize driving safety, string stability, and throughput through resource management.



\begin{figure*}
\centering
\includegraphics[width=.85\linewidth]{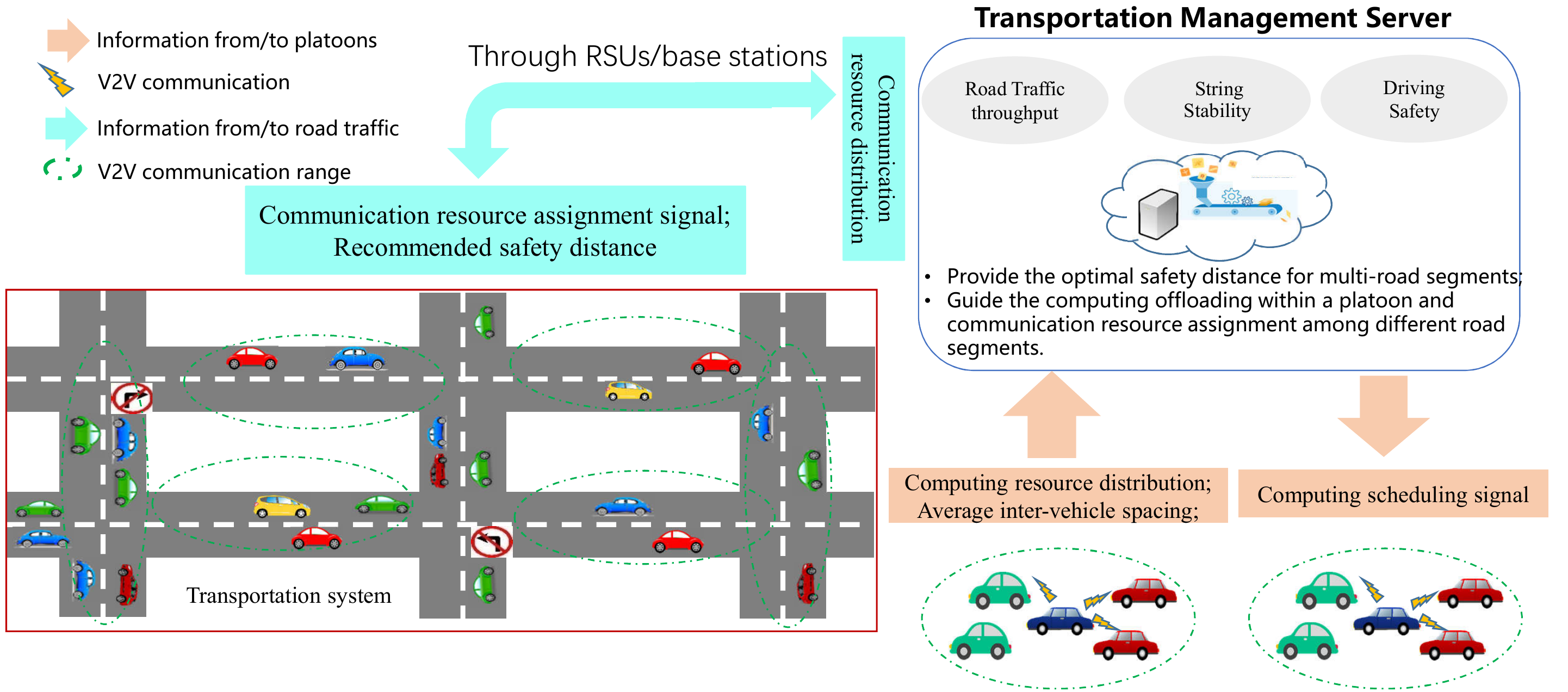} \hfill
\caption{Driving safety, string stability, and road traffic throughput are jointly optimized in multi-road segments. }
\label{gggj}
\end{figure*}
  

In order to realize efficient resource management in the V2V aided CAVs system, the critical work is to establish the quantitative relations between the resources and the road traffic metrics. There are many resource allocation studies \cite{{Qiao573},{Kato8847416}}, which attempt to establish reliable and efficient V2V resource sharing among vehicles. However, most of them did not give any mathematical explanation for the impact of resources on the road traffic. Based on the network calculus theory, Katsaros \textit{et al.} \cite{7277110} analyzed the upper bound of the end-to-end delay for the location-based routing in a hybrid vehicular network. However, they did not consider the multi-tasks scenario, which is impractical in CAVs systems. In addition, the aforementioned work also did not address the resource allocation with intermittent V2V communication due to the high mobility of vehicles. It is still an open challenge to jointly optimize driving safety, vehicle string stability, and road traffic throughput in the high dynamic CAVs system.





\section{System Model}



The system model of joint optimizing multiple road traffic metrics is illustrated in Fig.~\ref{gggj}. The average inter-vehicle spacing of each road segment is reported by platoons to the transportation management servers through the Road Side Units (RSUs) or cellular base stations. Each platoon will declare own computing capability. Additionally, the information about available bandwidth for a road segment is counted by the near RSUs or base stations.

By virtue of our proposed algorithm, the transportation management server provides the optimal safety distance (recommended inter-vehicle spacing) to each road segment to improve the string stability and road traffic throughput upon safe driving. However, optimal safety distance requires sufficient communication and computing resource. Because of the unbalanced resource distribution of the vehicles, the computing resource in some vehicles with powerful on-board processors are non-utilized. However, the resource in other vehicles is exhausted \cite{{Shao119}}. One possible solution to handle the maldistributed resource problem is the platoon-based edge computing, where a platoon of vehicles with sufficient on-board computing resources and communication bandwidth can offer additional mobile edge computing resources cooperatively \cite{{Wang8360847},{Kato8361406}}. Based on our proposed resource assignment algorithm, the transportation management server issues the instruction to implement the edge computing within a platoon and communication reassignment among road segments.

To achieve optimal driving experience, a joint utility of driving safety, string stability, and road traffic throughput should be addressed. Driving safety aims to avoid accidents occurring. The string stability typically is mainly impacted by the vehicle accelerations \cite{Dunbar5876300}. Road traffic throughput is determined by the average velocity and inter-vehicle spacing \cite{Nekoui2010Fundamental}. Moreover, these road metrics are all relevant to the inter-vehicle spacing. Hereafter, we investigate the continuum road traffic model to reveal the relation of inter-vehicle spacing with driving safety, string stability, and traffic throughput.







\subsection{Safety Distance}



The CACC system is used to make the follower keep a proper inter-vehicle spacing from the leader. This proper inter-vehicle spacing is named safety distance $s^*$. If the inter-vehicle spacing is less than the safety distance, the collision accident cannot be avoided \cite{Lian7322289}. To ensure the safe driving, we investigate the rear-end scenario, where the average velocity of vehicles is set to $\bm{v}$. The preceding vehicle starts to brake with the deceleration $A$ at time $t=0$. The follower detects the braking behavior of the preceding vehicle at $t=\tau_{01}$. Hereafter, it starts with the deceleration $A$ to slowdown at $t=\tau_{02} + \tau_{01}$, where $\tau_{01}$ represents the interval from the braking of the preceding vehicle to the perceiving of the follower. This interval depends on the V2V transmission delay. Moreover, $\tau_{02}$ is determined by the on-board processing capacity of the follower. To avoid accidents, the safety distance $s^*$ between two adjacent vehicles caters to Eq.~(\ref{asdskbkl}) at $t=0$,
\begin{equation}
s^*=\frac{A}{2}\tau_0^2+\boldsymbol{v}\tau_0,
\label{asdskbkl} 
\end{equation}

\noindent where $\tau_0 = \tau_{02}+\tau_{01}$ is referred to the perception-reaction delay or time, which is the duration of time from an accident happened to the driver reacts \cite{Nekoui2010Fundamental}. If the inter-vehicle spacing is larger than the safety distance, safe driving can be guaranteed.

\subsection{Vehicle String stability} 

There are many different aspects on stability. In this paper, we focus on the vehicle string stability that is characterized by the amplification of accelerations along the vehicle string. Moreover, stable string stability has the property that the maximal amplification of acceleration goes to 0 with the time approaching infinity \cite{Dunbar5876300}. It can be expressed as
\begin{equation}
\begin{split}
\max _{\mathbf{i}} \sup _{t}\left|\nabla^{2} u_{i}(t)\right|<\alpha \sup _{t}\left|\nabla^{2} u_{1}(t)\right|,
\end{split}
\label{kslflkasjklajlkklllll}
\end{equation} 

\noindent where $u_1(t)$ is the position of the head vehicle. $\nabla^{2} u_{i}(t)$ is the Laplace operator for the $i^{th}$ vehicle position, and $\alpha \in [0,1]$. Furthermore, there is $\nabla^{2} u_{i}(t) = c \nabla f$, where $f$ is the vehicle flux, and $c$ is a parameter that is proportional to $\nabla \rho$.

If Eq.~(\ref{kslflkasjklajlkklllll}) is violated, few accelerations of the vehicle in the vehicle string will result in so-called shockwaves of the vehicle string upon the dense road traffic condition. In the dense road traffic condition, the velocity of vehicles is lower than the speed limit on the road due to the short inter-vehicle spacing. While, in the sparse road traffic condition, the inter-vehicle spacing is large enough to make vehicles attain the speed limit without influencing road safety. The optimal driving strategy of a vehicle in the sparse road traffic condition is trivial: the velocity of the vehicle is equal to the road speed limit. Hence, in this paper, we only concentrate on the dense road traffic. 


In the dense road traffic, CAVs have to accelerate/decelerate continually to adjust the inter-vehicle spacing to approach the safety distance. When the inter-vehicle spacing is over the safety distance $s^*$, the follower will accelerate to narrow the gap that can improve the road traffic throughput. However, when the inter-vehicle spacing is less than the safety distance $s^*$, for safety driving, the follower should decelerate to leave enough inter-vehicle spacing with the preceding vehicle.  Then, we give a concise condition to guarantee the vehicle string stability upon the dense road traffic condition.

\begin{theorem} 
\label{1111xxx}
In the dense road traffic condition, as the safety distance $s^*$ approaching the current average inter-vehicle spacing $\frac{1}{\rho}$ of a vehicle string, i.e., 
\begin{equation}
\begin{split}
\|\frac{1}{\rho}-s^*\| \to 0,
\end{split}
\label{vjjdjdjjjjj}
\end{equation} 

\noindent the vehicle string stability is guaranteed.
\end{theorem}


%
%

\begin{proof}

Here, $\rho$ represents the vehicle density. We use $u_n(t)$ to denote the position of the $n^{th}$ vehicle in the vehicle string at the time $t$. Investigating the $n^{th}$ and $(n+1)^{th}$ vehicles of a vehicle string, if $(u_{n+1}-u_{n})>s^*$, the acceleration of $n^{th}$ vehicle is proportional to the gap $(u_{n+1}-u_{n})$. While $(u_{n+1}-u_{n})<s^*$, the deceleration of $n^{th}$ vehicle is inversely proportional to the gap $(u_{n+1}-u_{n})$. The behaviours of the $n^{th}$ and $(n-1)^{th}$ vehicles are same with that of $n^{th}$ and $(n+1)^{th}$ vehicles in the vehicle string. Moreover, the above statements in can be summarized as, 
\begin{equation}
\begin{split}
\frac{d^2 u_n}{d t^2} =-\zeta(2u_n-u_{n+1}-u_{n-1}),
\end{split}
\label{jhklgzxkhjvk}
\end{equation}

\noindent where $\zeta$ is a scale factor. $\frac{d^2 u_n}{d t^2}$ is the acceleration of the $n^{th}$ vehicle. Furthermore, the average position of the $n^{th}$ vehicle is written as $\left \langle  u(t)_n \right \rangle=\frac{1}{2}[u(t)_{n+1}+u(t)_{n-1}]$. Hence, we get 
\begin{equation}
\begin{split}
\frac{d^{2} u_{n}}{d t^{2}} =-2 \zeta\left(u_{n}-\left\langle u_{n}\right\rangle\right)=-2 \zeta\left(\frac{1}{\rho}-s^*\right).
\end{split}
\label{uhasdgfliaefphiua}
\end{equation} 


\noindent In addition, the basic conservation equation of the continuum road traffic model is given as \cite{Tosin2009},  
\begin{equation}
\frac{\partial \boldsymbol{v}}{\partial t}+(\boldsymbol{v} \cdot \nabla) \boldsymbol{v}= \frac{ d \boldsymbol{v}}{\rho dt},\label{dsafjh} 
\end{equation}

\noindent where $\boldsymbol{v}$ is the average velocity of vehicles on road. $t$ is the time component. According to $\frac{d^{2} u_{n}}{d t^{2}} = \frac{d \boldsymbol{v}}{dt}$, substituting Eq.~(\ref{uhasdgfliaefphiua}) to Eq.~(\ref{dsafjh}) and multiplying $\rho$ in the both sides, we obtain, 
\begin{equation}
\begin{split}
\rho\frac{\partial \boldsymbol{v}}{\partial t}+(\boldsymbol{v} \cdot \nabla) \boldsymbol{f}= 2 \zeta \left(\frac{1}{\rho}-s^*\right),
\end{split}
\label{sfhh}
\end{equation}     

\noindent in which $\boldsymbol{f}=\boldsymbol{v} \times \rho$ is the flux of the vehicle string. $\nabla \boldsymbol{f}$ represents the spatial gradient of the flux. If $\nabla \boldsymbol{f}$ is small, the spatial differentiation of $\boldsymbol{v}$ and $\rho$ will be constrained. Then, road traffic flow through different positions will become smooth that amplifies the driving experience. Similarly, $\frac{\partial \boldsymbol{v}}{\partial t}$ reflects the temporal differentiation of the vehicle velocity. The larger $\frac{\partial \boldsymbol{v}}{\partial t}$ indicates the heavy fluctuations of the vehicle string. Consequentially, the optimal string stability is attained when $\frac{\partial \boldsymbol{v}}{\partial t}$ and $\nabla \boldsymbol{f}$ are both equal to $0$. it results in 
\begin{equation}
\begin{split}
\|\frac{1}{\rho}-s^*\| = 0.
\end{split}
\label{adsgngfyti}
\end{equation}

\noindent Conversely, when the inter-vehicle spacing approaches the safety distance, i.e., $\|\frac{1}{\rho}-s^*\| \rightarrow 0$, the gradient of the vehicles flux $\nabla \boldsymbol{f} \rightarrow 0$. Furthermore, due to $\nabla^{2} u_{1}(t) = c \nabla f$, it implies $\max _{\mathbf{i}} \sup _{t}\left|\nabla^{2} u_{i}(t)\right| = \left|c \nabla f\right| \rightarrow 0$. 
Additionally, the flux of the preceding vehicle $\nabla^{2} u_{1}(t)$ is always positive. Therefore, the inequality $\max _{\mathbf{i}} \sup _{t}\left|\nabla^{2} u_{i}(t)\right|<\alpha \sup _{t}\left|\nabla^{2} u_{1}(t)\right|$ holds. The vehicle string stability is guaranteed.

\end{proof}

\subsection{Road Traffic Throughput}

The road traffic throughput $\mathscr{F}$ is defined as the scalar form of the vehicle string flux $\bm{f}$. Therefore, 
\begin{equation}
\begin{aligned}
\begin{split}
\mathscr{F} =  v \cdot \rho,
\end{split} 
\end{aligned} 
\label{vvvkjhvks}
\end{equation}

\noindent where the density $\rho$ is inversely proportional to the average inter-vehicle spacing. Thus, shorten the average inter-vehicle spacing that can increase road traffic throughput. Besides, due to applying the CACC system, the inter-vehicle spacing always surrounds the safety distance. Therefore, we can convert the maximization of road traffic throughput to minimize the safety distance $s^*$. However, the optimization of vehicle string stability aims to narrow the gap between safety distance and the average inter-vehicle spacing $\|\frac{1}{\rho} - s^*\|$. Consequently, transportation management cannot only optimize road traffic throughput but consider the impact on vehicle string stability.

\begin{remark}
In a dense road traffic condition, if the safety distance is less than the current average inter-vehicle spacing, i.e., $s^* < \frac{1}{\rho}$, improving road traffic throughput will deteriorate the vehicle string stability. 
\end{remark}


\section{Joint Optimization of Driving Safety, String Stability, and Road Traffic Throughput}

In this section, we propose an optimization model that jointly optimizes driving safety, string stability, and road traffic throughput in multi-road segments. Note that guaranteeing the driving safety is the prerequisite of optimizing string stability and road traffic throughput. To drive safety, the following vehicles should maintain the safety distance with the preceding vehicle in the vehicle string. In \cite{8644035}, the safety distance of road segment $i$ is determined by communication bandwidth $R_i$ and on-board computing capacity $\theta^i_e$ of vehicle $e$, where the computing capacity $\theta^i_e$ represents the CPU cycles of vehicle $e$ \cite{Kato8361406}. We assume vehicles in the same road segment competing with each other for the common V2V bandwidth. And, there are $M$ road segments in the transportation system. To achieve the optimal road traffic throughput and string stability upon driving safety, the optimization model is proposed as
\begin{equation}
\begin{aligned}
\begin{split}
 \textbf{P1:} \qquad &\min\limits_{R_i,\bm{\theta^i}} \ {\sum_{i}^{M} \left({s^*_i(R_i,\bm{\theta^i})}  + \delta \|s^*_i(R_i,\bm{\theta^i})-\frac{1}{\rho_i}\| \right) }\\
 \text{s.t.} \ \ \ \ \ 
&\text{C1:} \ \ s^*_i(R_i,\bm{\theta_i}) \ge 0, \ \ \text{C2:} \ \ 0 \le \sum_i R_i \le R^{upper}, \\
&\text{C3:} \ \ 0 \le \theta^i_e \le \theta_e^{upper},
\end{split} 
\end{aligned} 
\label{sdffff}
\end{equation}

\noindent where $s^*_i$ represents the safe distance in road segment $i$; $\bm{\theta^i} = \{\theta^i_1, \dots , \theta^i_g\}$ is the set of on-board computing resource offered by vehicles in road segment $i$. The vehicle density of the road segment $i$ is $\rho_i$; $\delta$ is a coefficient to balance road traffic throughput and string stability. $R^{upper}$ is the total communication resource of the transportation system. $\theta_e^{upper}$ is the available on-board computing capacity for vehicle $e$.

It is difficult to solve the $P1$ directly since $s^*_i(R_i,\bm{\theta^i})$ is nonlinear with $R_i$ and $\theta^i$. Therefore, to efficiently solve $P1$, we decompose $P1$ into two sub-problems in this paper. In the first part, we convert $P1$ into $P2$ that omits the communication and computing resources constraints on CAVs (constraint C2 and C3). The optimal safety distance of $P2$ jointly optimizes road traffic throughput and string stability without the resources constraints. Hereafter, the second part is designed to propose a resource management that allocates the communication and computing resource to meet the demands of the optimal safety distance. Moreover, the resource management caters to $C2$ and $C3$ constraints.

\begin{equation}
\begin{aligned}
\begin{split}
 \textbf{P2:} \qquad &\min\limits_{R_i,\bm{\theta^i}} \ {\sum_{i}^{M} \left({s^*_i }  + \delta \|s^*_i -\frac{1}{\rho_i}\| \right) }\\
 \text{s.t.} \ \ \ \ \ 
&\text{C4:} \ \ s^*_i  \ge 0  .
\end{split} 
\end{aligned} 
\label{sdf}
\end{equation}

\noindent In the first part, we introduce the intermediate constraint $s^* - z = 0$ to transform $P2$ to the general form to apply the consensus Alternating Directions Method of Multipliers algorithm (ADMM) \cite{BoydS2010}. One benefit of the consensus ADMM is to make the safety distance of different road segments identical by iterations, which gradually eliminates the traffic fluctuations when vehicles alter into another road segment. In addition, the monotonicity of $\|s^*\|^2_2$ is same with that of $s^*$. Thus, we can replace $s^*$ with $\frac{1}{2}\|s^*\|^2_2$ in $P3$ to transform $P2$ to the standard LASSO form \cite{{BoydS2010}}.

\begin{equation}
\begin{aligned}
\begin{split}
 \textbf{P3:} \qquad &\min\limits_{\bm{s^*}} \ { \frac{1}{2}\sum_{i}^{M} \left\|s^*_i \right\|_2^2     + \delta \|z - \frac{1}{\rho_i}\| }\\
 \text{s.t.} \ \ \ \ \ 
&\text{C5:} \ \ s^*_i - z = 0.
\end{split} 
\end{aligned} 
\label{gfg}
\end{equation}


\noindent  Since the minimum $ { \frac{1}{2}\sum_{i}^{M} \left\|s^*_i \right\|_2^2 + \delta \|z - \frac{1}{\rho_i}\| }$ results in a non-negative $s^*$, the constraint $C4$ is unnecessary in $P3$. Moreover, the augmented Lagrangian of $P3$ is\begin{equation}
\begin{aligned}
\begin{split}
\begin{array}{l}{\mathcal{L}_{\mu}\left({s^*}_{1}, \ldots, {s^*}_{M}, \bm{y}, z\right)} \\ {=\sum\limits_{i=1}^{M} \frac{1}{2}\left\|s^*_i \right\|_2^2 + {y}_{i} ^{T}\left(  s^*_i-z\right)+\frac{\mu}{2}\left\| s^*_i-z\right\|_{2}^{2}} + \delta \|z-\frac{1}{\rho_i}\| \end{array},
\end{split} 
\end{aligned} 
\label{sfddkjnh}
\end{equation}

\noindent where $\mu \in R_{+}$ is a positive penalty parameter in the augmented Lagrangian \cite{Zhou8667693}. $\bm{y}$ is the vector of Lagrange multipliers. The iterations of updating $s^*_i$, global variable $z$, and Lagrange multiplier $\xi_i$ are given as
\begin{equation}
\begin{aligned}
\begin{split}
\begin{aligned} {s^*_i}^{k+1} & :=\left(1+\mu\right)^{-1} \mu \left(z^k-\xi_i^{k} - \mathbb{E}\left\{ \boldsymbol{\frac{1}{\rho}} \right\} \right) \\ \
z^{k+1} & :=S_{\delta / \mu}\left(\mathbb{E}\left\{\boldsymbol{{s^*}}^{k+1}+\boldsymbol{\xi}^{k}\right\}\right) + \mathbb{E}\left\{ \boldsymbol{\frac{1}{\rho}} \right\} \\ \
\xi_i^{k+1} & :=\xi_i^{k}+{s^*_i}^{k+1}-z^{k+1}, \end{aligned}
\end{split} 
\end{aligned} 
\label{vvvvvv}
\end{equation}

\noindent where $\xi_i  =  \frac{y_i}{\mu}$. $\mathbb{E}\{a\}$ is the expectation of $a$. $\boldsymbol{s^*}=\{s^*_1,\dots, s^*_M\}$ is the set of safety distance in different road segments. $\boldsymbol{\xi}$ represents the set of $\{\xi_1,\dots , \xi_M\}$. $\boldsymbol{\frac{1}{\rho}} =\{\frac{1}{\rho_1}, \dots , \frac{1}{\rho_M} \}$. $S_{{{\delta}/{\mu}}}(a)$ is the soft threshold operator that is given by \cite{BoydS2010}
\begin{equation}
\begin{aligned}
\begin{split}
S_{{{\delta}/{\mu}}}(a)=\left\{\begin{array}{ll}{a-{{\delta}/{\mu}}} & {a>{{\delta}/{\mu}}} \\ {0} & {|a| \leq {{\delta}/{\mu}}} \\ {a+{{\delta}/{\mu}}} & {a<-{{\delta}/{\mu}}.}\end{array}\right. 
\end{split} 
\end{aligned} 
\label{vvvdfsfsdh}
\end{equation}

\noindent In addition, the stopping criterion of the iteration is
\begin{equation}
\begin{aligned}
\begin{split}
\left\|{r}^{k}\right\|_{2}^{2} &=\sum_{i=1}^{M}\left\|{s^*_{i}}^{k}-z^{k}\right\|_{2}^{2}\le \epsilon^{\text { prim }} 
\\ \qquad\left\|{Dr}^{k}\right\|_{2}^{2} &=M \mu^{2}\left\|{z}^{k}-{z}^{k-1}\right\|_{2}^{2} \le \epsilon^{\text { dual }},
\end{split} 
\end{aligned} 
\label{fwgljhewrglkj}
\end{equation}

\noindent where ${r}^{k}$ and ${Dr}^{k}$ represent the primal residual and the dual residual, respectively \cite{BoydS2010}. $\epsilon^{\text { prim }}$ and $\epsilon^{\text { dual }}$ are the constant thresholds. Based on the above analysis, the multi-road segment joint optimization scheme is summarized as Alg.~{\ref{Ag1}}.


\begin{algorithm}  
  \caption{ADMM-based joint optimization scheme}  \label{Ag1}
 
  $\bf{input:}$ $\boldsymbol{\frac{1}{\rho}}$, $z \leftarrow 1$, $\xi_i \leftarrow \frac{y_i}{\mu} $  \;  

  \While{ $\left\|{Dr}^{k}\right\|_{2}^{2} \leq \epsilon^{\text { dual }}$ and  $\left\|r^k_{i}\right\|_{2}^{2} \leq \epsilon^{\text { prim }}$}   
  {   

      Update ${s^*_i}^{k+1} \leftarrow \left(1+\mu\right)^{-1} \mu \left(z^k-\xi_i^{k} - \mathbb{E}\left\{ \boldsymbol{\frac{1}{\rho}} \right\}\right)$ \;
          
      Update $z^{k+1} \leftarrow S_{\delta / \mu}\left(\mathbb{E}\left\{\boldsymbol{{s^*}}^{k+1}+\boldsymbol{\xi}^{k}\right\}\right) + \mathbb{E}\left\{ \boldsymbol{\frac{1}{\rho}} \right\}$ \;

      Update $\xi_i^{k+1} = \xi_i^{k}+{s^*_i}^{k+1}-z^{k+1}$ \;
      Update $r_i^{k+1} = {s^*_i}^{k+1} - z^{k+1}$ \;   
    
  }    
  $\bf{output:}$ $\boldsymbol{s^*}$ 
\end{algorithm}

\section{Upper Bound of V2V Offloading Delay}

To take account of constraints $C2$ and $C3$, we need to figure out the demands of communication and computing resource to support the optimal safety distance. In addition, a typical task offloading paradigm includes the communication process and the computing process. Therefore, the V2V task offloading performance can be used to
determine the demands of communication and computing resource in the CACC system. However, most previous works mainly concentrated on the access delay \cite{{5967982}}. As for edge computing-based applications, the computing/processing delay is also a dominant factor to affect the quality of services. Hence, we take the account of the communication delay and the computing delay in the delay performance analysis.

In addition, this paper focuses on the upper bound delay performance rather than the average delay performance. If we adjust the upper bound delay of V2V offloading smaller than the delay requirement of the application, the delay requirement of the application can be guaranteed. However, the average delay is an average time consuming of the V2V offloading. When the variance of the offloading delay becomes large, the average delay cannot give any promise to complete the application offloading in time. The upper bound delay of vehicular communication have been widely studied \cite{5967982}. However, the previous work did not consider the transmission collisions. An advantage of the network calculus (NC) is that it can easily obtain the end-to-end upper bound delay in the competitive concatenated system. Thus, in this paper, we apply the NC model to obtain the end-to-end delay for the V2V application offloading.

In the NC theory, an arrival process $A(s,t) = A(t) - A(s)$ represents the cumulative number of the input network traffic of a vehicle in the time interval $(0,t]$. There are $K$ categories of driving assistance applications $\{1,\dots,K\}$. The data volume of the CACC application is denoted by $o_k$ and the average arrival rate is $\lambda_k$ \cite{Kato7636965}. According to the $(\epsilon, \sigma)$-upper constraint, we have $A(t)-A(s) \le \lambda_k(t-s)+o_k$ \cite{Jiang2008Stochastic}. Thus, the arrival curve of CACC application in vehicle $i$ is $\alpha_{i,k}(t)=\lambda_k\cdot t+o_k$.

Besides, the total volume of network traffic should not exceed the communication capacity, i.e., $N\sum_{i=1}^{K} \lambda_i o_i\leq R_j$. $N$ is the number of CAVs in road segment $j$. $R_j$ is the communication bandwidth of road segment $j$. Based on the optimal safety distance obtained from Alg.~{\ref{Ag1}}, we get $N = \frac{2\iota L}{s^*}$, where $\iota$ is the number of lanes in the road segment. $L$ is the V2V radio range. Hereafter, the upper bound delay of V2V offloading is given as follow

\begin{theorem}

\label{asdkjgsl}
The upper bound delay of V2V offloading for vehicle $i$ with application $k$ in road segment $j$ is
\begin{equation}
\begin{split}
T_{(ij)k} = \frac{o_k\eta_k}{\theta_i} + \frac{o_k}{R_j - H_{\lambda}}+ \frac{\Lambda H_{\lambda}+ H_o}{R_j - H_\lambda}+\Lambda .
\end{split}
\label{dfhyssssssssssssssssssssssssjbh}
\end{equation}

\noindent where $\Lambda = \left(2^{\varepsilon+1}-1+2^{\varepsilon}(\gamma-\varepsilon)\right)W_0$ is the protocol-related part; $\frac{o_k\eta_k}{\theta_i}$ is the computing delay; $\frac{o_k}{R_j - H_{\lambda}}$ represents the transmission delay; $\frac{\Lambda H_{\lambda}+ H_o}{R_j - H_\lambda}$ is regarded as the competition delay, in which $H_{\lambda}=N\sum_{l\neq k}^{K}\lambda_l+(N-1)\lambda_k$ and $H_{o}=N\sum_{l\neq k}^{K} o_l+(N-1)o_k$.

\end{theorem}

\begin{proof}

The competitive V2V communication applies the contention-based medium access control approaches such as the IEEE 802.11p standard, which resorts to the exponential back-off algorithm. This paper assumes the exponential back-off process has $\gamma$ back-off states and the initial size of the back-off window is $W_0$. Therefore, the size of the window in the back-off state $g$ is expressed as $W_g=\min\left\{2^{{g}} W_0,2^{\varepsilon}W_0\right\}$, where $\varepsilon$ is a threshold to limit the increase of the counter, $0 < \varepsilon \le \gamma$. If the back-off counter exceeds $\varepsilon$, the size of the back-off window will not grow anymore. Therefore, the maximum waiting time for the V2V access is
\begin{equation}
\begin{split}
\sum_{g=0}^\gamma W_g &=\left(2^{\varepsilon+1}-1+2^{\varepsilon}(\gamma-\varepsilon)\right)W_0.
\end{split}
\label{54186122}
\end{equation}

According to the superposition property \cite{Jiang2008Stochastic}, the whole arrival curves except for the CACC application traffic of vehicle $i$ is regarded as a superposition curve $\alpha_{ik}(t)$:
\begin{equation}
\begin{split}
&\hat{\alpha}_{ik}(t)=\sum_{j\neq i}^N\sum_{l\neq k}^{K}\alpha_{j,l}(t)\\&=
\left(N\sum_{l\neq k}^{K}\lambda_l+(N-1)\lambda_k\right)\cdot t+N\sum_{l\neq k}^{K} o_l+(N-1)o_k.
\end{split}
\label{dmhgkf;lyi}
\end{equation}


The transmission capability of the competitive V2V channel is constrained by the classical latency-rate service curve $\beta(t)=R(t-x)^+$ \cite{Jiang2008Stochastic}, in which $x$ is the service delay of the V2V channel. Hence, we get
\begin{equation}
\begin{split}
\beta(t)=R_j \left(t-\left(2^{\varepsilon+1}-1+2^{\varepsilon}(\gamma-\varepsilon)\right)W_0\right)^+,
\end{split}
\label{sdg;kj}
\end{equation}

\noindent where $(x)^+=\max\{x,0\}$. Next, according to the theory of Leftover Service \cite{Jiang2008Stochastic}, we obtain the service curve of the V2V transmission to serve CACC application $k$ for vehicle $i$
\begin{equation}
\begin{aligned}
\begin{split}
\beta_{i,k}(t)&=(\beta-\hat{\alpha}_{ik})^+(t)\\
&= (R_j - H_{\lambda}) \left[ t - \left( \frac{\Lambda   H_{\lambda}+H_o}{R_j - H_{\lambda}}+\Lambda   \right)\right]^+,
\end{split} 
\end{aligned}
 \label{rty;lyi} 
\end{equation}

\begin{figure}
\centering
\subfloat[Upper bound delay of the different communication and computing capacities in a platoon with 3 vehicles.]{\label{skskskskskkkks}{\includegraphics[width=0.48\linewidth]{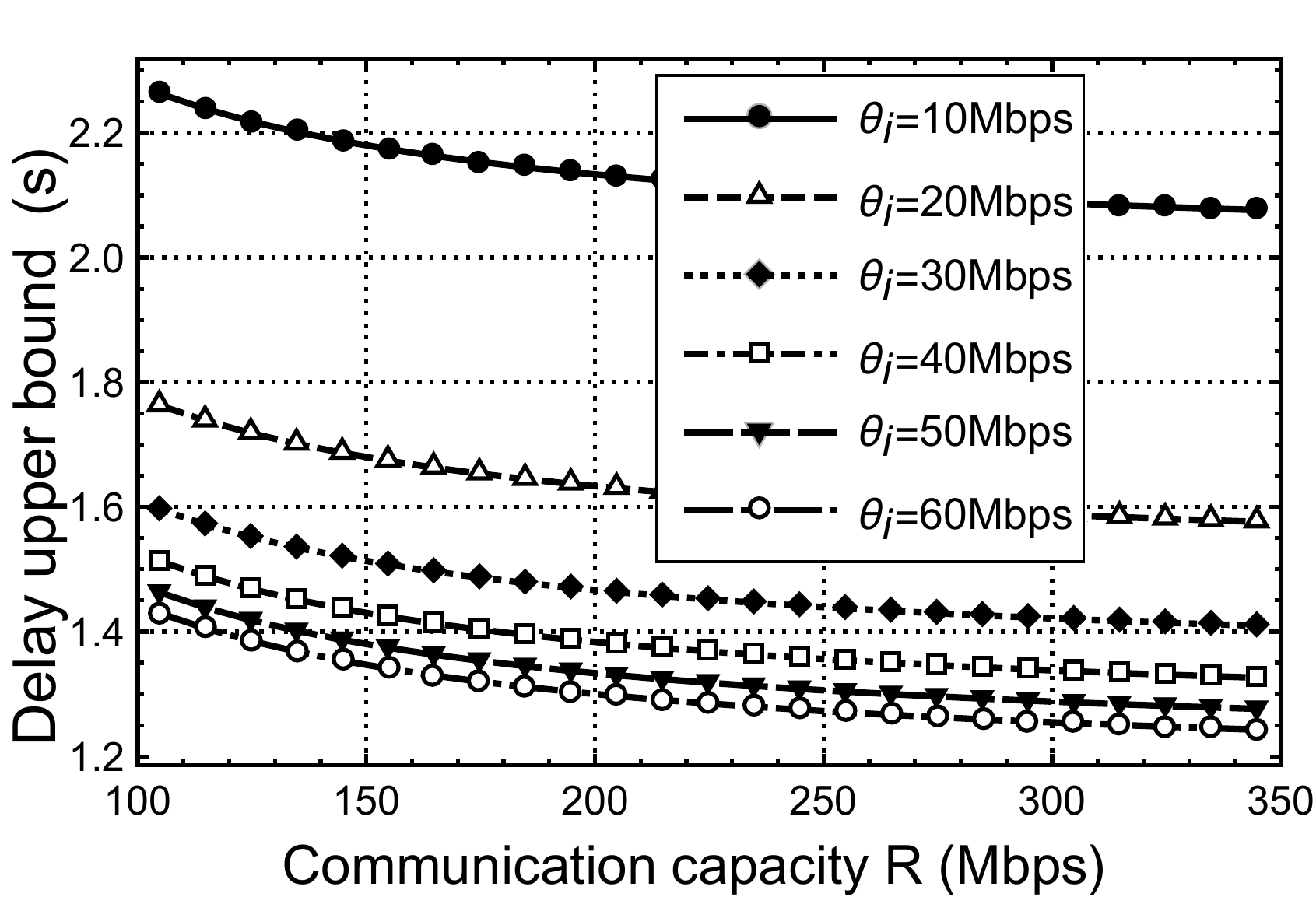}}}\hfill
\subfloat[Upper bound delay of the different number of vehicles, where the computing capacity of each vehicle is 60 Mbps.]{\label{jhadsjkdskjjkhb_nmsajd}
{\includegraphics[width=0.48\linewidth]{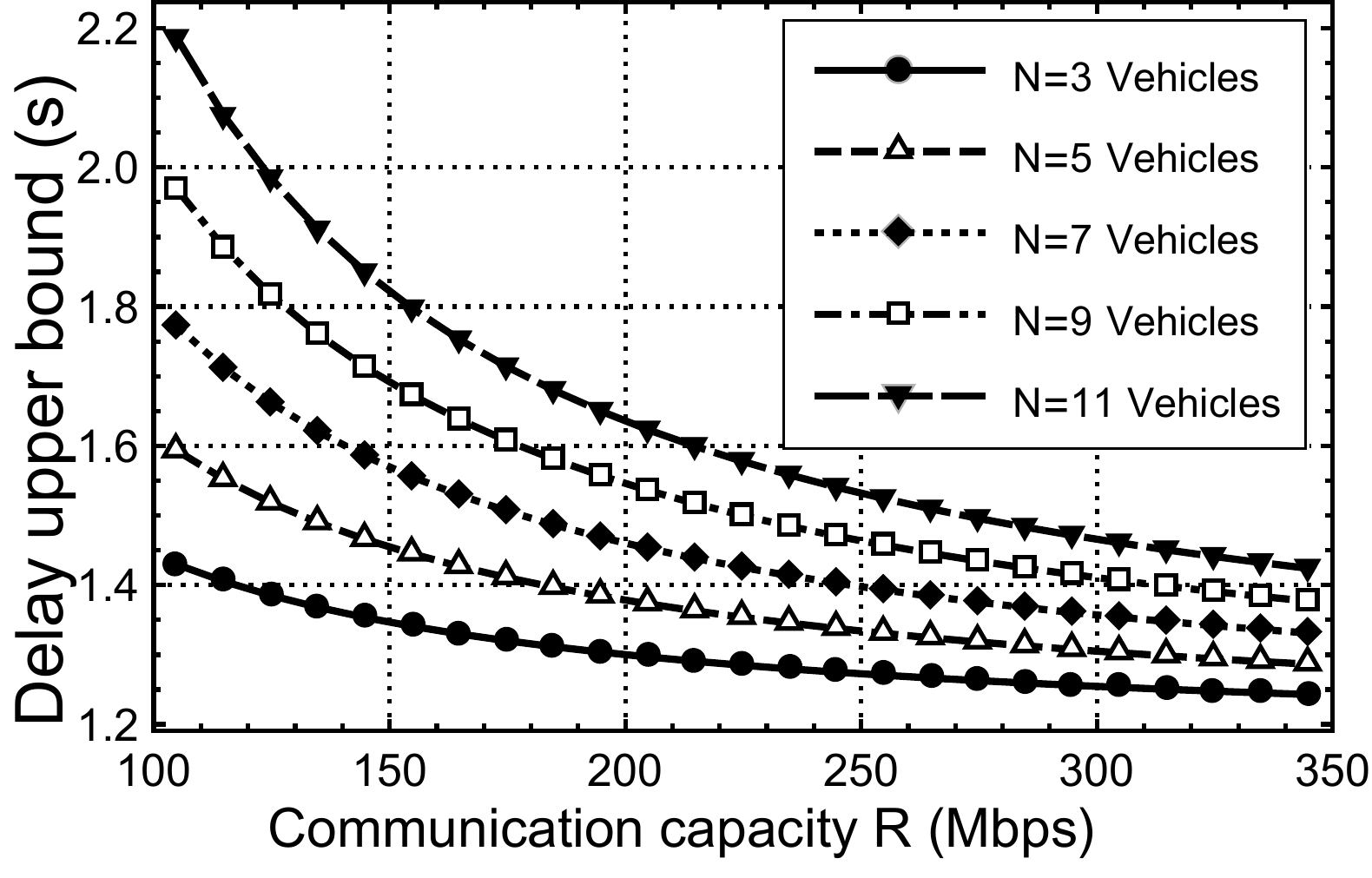}}} \hfill
\caption{Numerical demonstration of Eq.~({\ref{upperbound}}).}
\label{fig_NC_Performance}
\end{figure}

\noindent where $\Lambda = \left(2^{\varepsilon+1}-1+2^{\varepsilon}(\gamma-\varepsilon)\right)W_0$, $H_{\lambda}=N\sum_{l\neq k}^{K}\lambda_l+(N-1)\lambda_k$ and $H_{o}=N\sum_{l\neq k}^{K} o_l+(N-1)o_k$. Similarly, the service curve of on-board computing and executing is
\begin{equation}
\begin{split}
\beta_{i,k,h}(t)=r_h\left(t-\frac{o_k\eta_k}{\theta_i}\right)^+ ,
\end{split}
\label{sdbb}
\end{equation}

\noindent where $r_h$ is the execution rate, which is much higher than the V2V transmission capacity, i.e. $r_h \gg R$. Then, according to the concatenation property \cite{Katsaros2016End}, the total offloading service curve of the CACC application for vehicle $i$ is
\begin{equation}
\begin{split}
\beta_{i,k}^{j}(t)=\beta_{i,k}\otimes\beta_{i,k,h},
\end{split}
\label{jkljghkpyu}
\end{equation}

\noindent where $(a\otimes b)(x)=\inf_{0\leq y\leq x}[a(y)+b(x-y)]$. Thus, we get
\begin{equation}
\begin{split}
\beta^j_{i,k}(t) = (R_j - H_{\lambda}) \left[ t - \left( \frac{o_k \eta_k}{\theta_i}+\frac{\Lambda   R_j+H_o}{R_j - H_{\lambda}} \right)\right]^+.
\end{split}
\label{lkjhlwm}
\end{equation}


\noindent Based on the delay bound theorem \cite{Jiang2008Stochastic}, the service delay $\mathscr{D}_{i,k}(t)$ of the application offloading satisfies
\begin{equation}
\begin{split}
\mathscr{D}_{i,k}(t)\leq h(\alpha_{i,k}(t),\beta_{i,k}^j(t)),
\end{split}
\label{xmxmxm}
\end{equation}

\noindent where $h(a,b)=\sup_{h\geq 0}\{\inf\{\tau\geq 0:a(x)\leq b(x+\tau)\}\}$. Finally, the upper bound of the V2V offloading delay of application $k$ for vehicle $i$ is
\begin{equation}
\begin{split}
T_{(ij)k} &= h(\alpha_{i,k}(t),\beta_{i,k}^j(t)) \\
&=\frac{o_k\eta_k}{\theta_i} + \frac{o_k}{R_j - H_{\lambda}}+ \frac{\Lambda H_{\lambda}+ H_o}{R_j - H_\lambda}+\Lambda .
\end{split}
\label{upperbound}
\end{equation}


\end{proof}

To show the effectiveness of the proposed bound in Eq.~(\ref{upperbound}), we use a simple numerical example to demonstrate the impact of communication and computing resources on the upper bound delay of V2V offloading in Fig.~\ref{fig_NC_Performance}. Fig.~\ref{skskskskskkkks} depicts the curve of upper bound delay $T_{(ij)k}$ with the communication capacity $R_j$ and computing capacity $\theta_i$ of vehicle $i$ in road segment $j$. In the numerical scenario, the platoon is assembled by $3$ vehicles. Each vehicle has to support $5$ vehicular assistance applications, i.e. $K=5$. The data volume of applications $\bm{o} = \{o_1,\dots,o_K\}$ is randomly distributed from $[1,3]$ (Mb). The arrival rate $\bm{\lambda} = \{\lambda_1,\dots,\lambda_K\}$ is generated from the uniform distribution where $U$~$[0.4, 0.8]$. $W_0 = 0.2$, $\eta_k = 5$, $\gamma = 2$, and $\varepsilon = 1$. As shown in Fig.~\ref{skskskskskkkks}, the upper bound delay drops with the rising communication capacity and on-board computing capacity $\theta_i$. While $\theta_i$ is generated from $10$ to $20$, the upper bound declines significantly. However, if $\theta_i$ is over $40$, the benefits from the high computing will become less. This is because the bottleneck of the upper bound delay is caused by the V2V transmission rather than computing delay when a vehicle has enough computing capacity.

Fig.~{\ref{jhadsjkdskjjkhb_nmsajd}} demonstrates the upper bound delay with a different number of vehicles. Due to the competition among vehicles, the upper bound delay increases with the number of vehicles. However, when the bandwidth is sufficiently large, then the delay caused by the transmission competition can be negligible. Therefore, upon the large communication bandwidth, the upper bound delay is less affected by the number of vehicles in the platoon, but it is predominated by the computing capacity of each vehicle.


\begin{lemma}
\label{kjasdfksal}
When the on-board computing is sufficient large, then the upper bound of V2V offloading delay is reduced to the summation of a transmission delay, a competition delay, and a protocol-related part. i.e., $\lim\limits_{\theta_i \to +\infty}T_{(ij)k} = \frac{o_k}{R_j - H_{\lambda}}+ \frac{\Lambda H_{\lambda}+ H_o}{R_j - H_\lambda}+\Lambda$. While the bandwidth is sufficient large, then the upper bound of V2V offloading delay is reduced to the computing delay plus the protocol-related part. i.e., $\lim\limits_{R_j \to +\infty}T_{(ij)k} = \frac{o_k\eta_k}{\theta_i}+\Lambda$.
\end{lemma}

\begin{proof}
The above equations can be derived by Eq.~({\ref{upperbound}}) taking $\theta_i\to \infty$ and $R_j\to\infty$, respectively.
\end{proof}

Lemma \ref{kjasdfksal} indicates the efficient way to increase resource that can significantly reduce the upper bound delay of V2V offloading.

\section{Multi-Armed Bandit Resource Scheduling} 

The CACC system needs sufficient computing and communication to process diverse and historical kinetic analyses to maintain the optimal inter-vehicle spacing among vehicles. However, many other automated assistance applications, such as cooperative malicious attacks detection applications \cite{Wang7999188}, and cooperative lane change applications, will compete with the CACC application for the limited bandwidth and on-board process capacity. If the CACC system cannot obtain sufficient resources to maintain the optimal safety distance obtained from Alg.~{\ref{Ag1}}, CAVs will increase the inter-vehicle spacing to reduce the resource demands of the CACC application \cite{8644035}. 

In general, different vehicles occupy different computing resources since the on-board processors are diverse. Besides, the available V2V communication bandwidth of a road segment is determined by the number of vehicles and the bandwidth assignment of the road segment \cite{8080373}. Because of the unbalanced distribution of vehicles and resources, some vehicles with the deficient computing or communication resource cannot attain the optimal inter-vehicle spacing.

In this section, we study the resource allocation under intermittent V2V communication. Here, through Alg.~\ref{Ag1}, we can obtain the optimal safety distance $s^*$ of $P3$. Hereafter, according to Eq.~(\ref{asdskbkl}), the delay requirement of CACC application to maintain the optimal $s^*$ among CAVs is 


\begin{equation}
\begin{aligned}
\begin{split}
\tau_0(s^*) = \frac{\sqrt{v^2+2As^*} - v}{A} .
\end{split} 
\end{aligned} 
\label{cjscnsj}
\end{equation}

\noindent where $\tau_0$ is the perception-reaction delay that represents the duration from an event happens to the preventive action adopted by vehicles \cite{Nekoui2010Fundamental}. To satisfy the limited resource constraints, the upper bound of the V2V offloading delay $T_{(ij)k}$ should not exceed the perception-reaction delay $\tau_0(s^*)$ of the optimal safety distance $s^*$.

Therefore, the vehicles in a road segment can be divided into two groups: one group is resource-deficient vehicles, another is resource-rich vehicles. The criterion of distinguishing the two groups is based on the value of $(T_{(ij)k} - \tau_0)$. Vehicles with $(T_{(ij)k} - \tau_0) <0$ are clustered into $J^1$ that represents the vehicles with sufficient resources. While vehicles with $(T_{(ij)k} - \tau_0) > 0$ are clustered into $J^0$ that represents the vehicles lacking of resource. Due to the resource constraints, the inter-vehicle spacing of vehicles in $J^0$ cannot maintain the optimal safety distance with the preceding vehicle.

Hereafter, we rank $J^0 = \{J^0_1, J^0_2, \dots J^0_p \}$ in order of descending $(T_{(ij)k} - \tau_0)$. For instance, $J^0_1$ is the vehicle with the largest $(T_{(ij)k} - \tau_0)$. Next, $J^0_1$ will offload its application to the vehicles in set $J^1= \{J^1_1, J^1_2, \dots J^1_q \}$. Afterwards, vehicle $J^0_2$ offloads its applications to the $J^1$ vehicles, and so on. Before each offloading, vehicles will update the value $(T_{(ij)k} - \tau_0)$ for the next scheduling. In addition, we draw the offloading pairs of vehicles of $J^0$ and $J^1$ in Fig.~\ref{cvcv}. In this demonstration, each vehicle of $J^0$ has applications needed to offload. The offloading targets are selected from $J^1$. Each vehicle of $J^1$ can handle multiply offloading tasks according to the redundant computing capability.

\begin{figure} 
     \centering
     \includegraphics[width=0.4\textwidth]{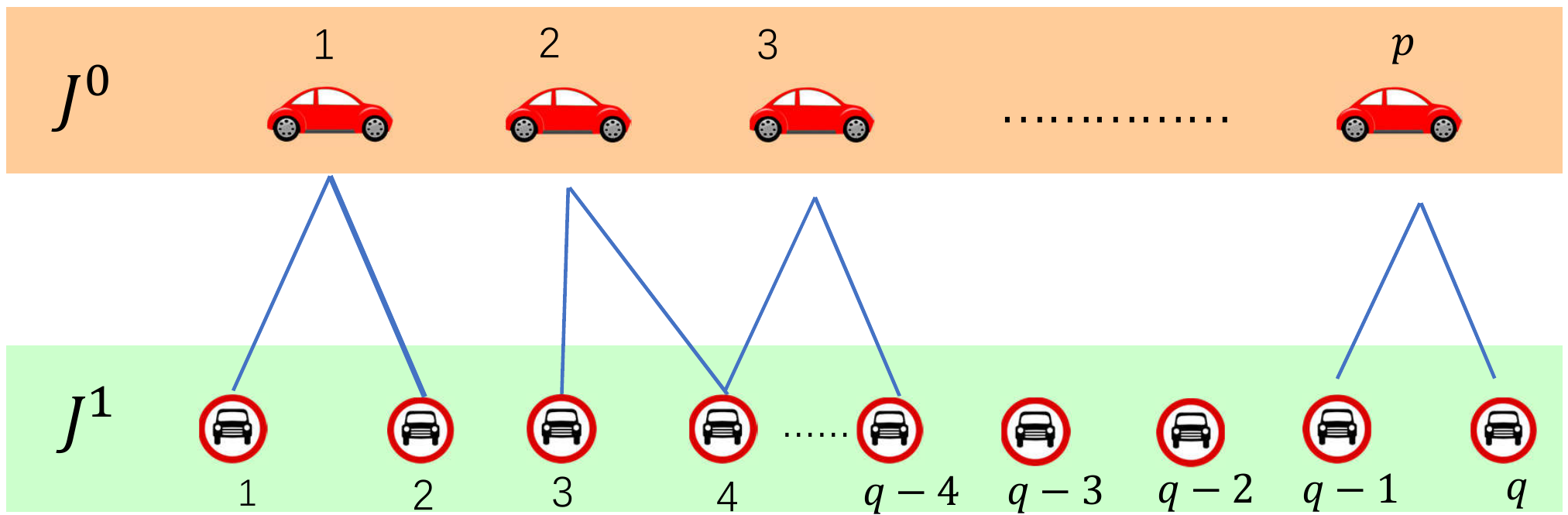} 
     \caption{Offloading pairs between $J^0$ and $J^1$.} 
     \label{cvcv}
\end{figure} 


Consequently, we propose a Sleeping Multi-Armed Bandit Tree-based Offloading (SMTO) scheduling that selects candidate vehicles to process the offloading applications in a high mobility circumstance. The sleeping Multi-Armed Bandit (MAB) model refers to the sequential optimization problem where the action set is time-varying \cite{Kleinberg2010}. The available actions at each round are uncertain that is same with the intermittent V2V transmission, where the application will disappear when the vehicle leave the radio range of the target vehicle. 

We assume that the offloading vehicle does not have prior knowledge about its environment (e.g., which vehicles will leave or stay in the radio range). This significantly reduces the communication overhead. The offloading vehicle does not need to issue its kinetic information to surroundings. Our proposed algorithm focuses on the vehicle platoon, where vehicles may drive off the communication platoon during the offloading. The candidate offloading target are selected by
\begin{equation}
\begin{aligned}
\begin{split}
b_{ig}(t+1)&=\underset{j \in \mathcal{N}_{i}}{\arg \max } \ Q_g(t)\\&+ \sqrt{\frac{\mathscr{P}_g\left[\tau_{k}-T_{(i j) k}\right]^{+} \ln n_{(i j)}(t)}{\mathcal{J}_{(i j)}(t)}} ,
\end{split} 
\end{aligned} 
\label{rrrrrr}
\end{equation}

\noindent where $b_{ig}(t)$ represents the selected target to process the offloading application $g$ from the vehicle $i$ at time $t$. $Q_g(t)$ is the reward of finishing the application $g$ at time $t$ \cite{Kato8657791}. $n_{(ij)}(t)$ is the connected duration between the vehicle $i$ and vehicle $j$ at time $t$, which can be measured by the period HELLO message. When vehicle $j$ leaves the platoon, $n_{(xj)}(t)$ is reset to $0$, where $x$ represents any vehicle in the platoon. $\mathcal{J}_{(ij)}(t)$ is the number of selection times of vehicle $j$ to be the offloading target for vehicle $i$. And, $\mathscr{P}_g$ is a weight of application $g$. However, according to Eq.~(\ref{rrrrrr}), if a new vehicle appears in the platoon, it will be selected as the offloading target. The reason is that a new vehicle usually stays longer than the previous vehicles in the platoon. So that the new vehicle can provide a more stable V2V connection than that of the previous vehicles in the platoon. 



The Sleeping MAB Tree search structure is illustrated in Fig.~\ref{dfvjvdvvv}. Since one platoon can only support $N = \frac{2\iota L}{s^*}$ vehicles, each vehicle connects with $N-1$ vehicles via V2V communication. There are $K$ number of V2V applications for offloading. In addition, we sort the vehicular applications with the order of priority, where application $i$ represents the $i^{th}$ priority application. The application $1$ has the highest priority. The application with high priority is delay-sensitive, such as CACC and lane change assist applications. As shown in Fig.~\ref{dfvjvdvvv}, the first row of the tree demonstrates the application $1$ offloading. Subsequently, the application $2$ is offloaded in the second row, etc. In each application offloading, the algorithm will check the $J^0_j$ of the road segment $j$ is whether or not empty. If $J^0_j\in \emptyset$, the algorithm is stopped. 

\begin{figure} 
     \centering
     \includegraphics[width=0.4\textwidth]{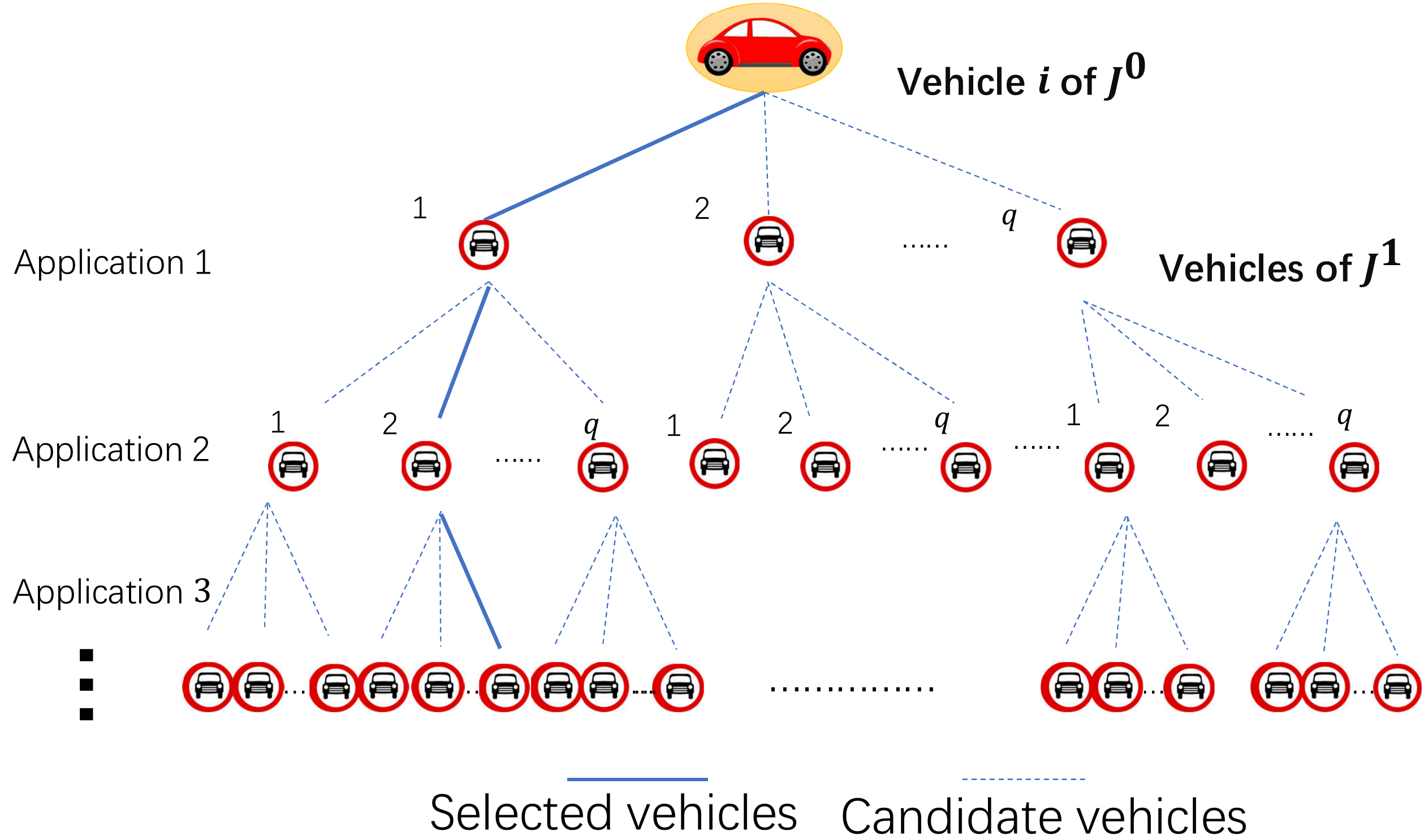} 
     \caption{Sleeping MAB Tree-based Offloading.} 
     \label{dfvjvdvvv}
\end{figure}

In this paper, the V2V bandwidth of each road segment is centrally assigned by the transportation management server. After the match of $J^0$ and $J^1$, we divide the road segments into two groups. The road segments with $J^0 = \emptyset$ belong to $J^0_{empty}$. The other road segments with $J^0 \neq \emptyset$ belong to $J^0_{exist}$, where $J^0_{empty}$ represents the road segment without any deficiency vehicle. However, $J^0_{exist}$ represents the road segments that still have some vehicle lacking of resource to maintain the optimal safety distance. According to Eq.~({\ref{upperbound}}), the minimum recouped communication bandwidth $\mathscr{R}^{exist}_j$ is used to recoup the bandwidth of road segment $j\in J^0_{exist}$ that is identical to
\begin{equation}
\begin{split}
\mathscr{R}^{exist}_j = \max_{i \in j} \left[\frac{o_k + \Lambda H_{\lambda}+ H_o}{\tau_0 - \frac{o_k\eta_k}{\theta_i} - \Lambda} + H_{\lambda} - R_j \right].
\end{split}
\label{upperchangi}
\end{equation}

\noindent On the other hand, road segments of $J^0_{empty}$ provide their part of communication bandwidth to recoup the deficient resource in $J^0_{exist}$. Similarly, based on Eq.~({\ref{upperbound}}), the maximum communication bandwidth $\mathscr{R}^{empty}_u$ supplied by road segment $u \in J^0_{empty}$ is
\begin{equation}
\begin{split}
\mathscr{R}^{empty}_u = \min_{i \in u} \left[R_u - \frac{o_k + \Lambda H_{\lambda}+ H_o}{\tau_0 - \frac{o_k\eta_k}{\theta_i} - \Lambda} - H_{\lambda} \right].
\end{split}
\label{upaaaapercwwwww}
\end{equation}

\noindent Therefore, to balance the bandwidth distribution of different road segments, the provided bandwidth of road segment $u \in J^0_{empty}$ is 
\begin{equation}
\begin{split}
\Delta R^{empty}_u = \mathscr{R}^{empty}_u - \max \left[ \frac{D_\mathscr{R}}{M}, \ \ 0 \right].
\end{split}
\label{xslxlskxlsxk}
\end{equation}

\noindent And, the supplied bandwidth of road segment $j\in J^0_{exist}$ is
\begin{equation}
\begin{split}
\Delta R^{exist}_j = \mathscr{R}^{exist}_j + \frac{D_\mathscr{R}}{M}.
\end{split}
\label{xddldxxx}
\end{equation}

\noindent where $D_\mathscr{R} = \sum\limits_{u\in J^0_{empty}}{\mathscr{R}^{empty}_u} - \sum\limits_{j\in J^0_{exist}}{\mathscr{R}^{exist}_j}$ is the difference between the total surplus bandwidth of $J^0_{empty}$ and the total deficient bandwidth of $J^0_{exist}$. While $D_\mathscr{R} < 0$, the total bandwidth of the transportation system cannot maintain all vehicles with the optimal safety distance $s^*$. Some vehicles in the deficient road segment will increase their average inter-vehicle spacing (sacrificing road traffic efficiency) to guarantee driving safety.

The process of the SMTO resource allocation is elaborated in Alg.~\ref{Ag2}, where $\bm{R} = [R_1,\dots, R_M]$. $M$ is the total number of road segments. $Q_{parent \ of \ b_{ig}}$ is the reward of the parent of the node $g$ in vehicle $i$. At time $t$, if a new vehicle becomes available to connect with vehicle $i$, the algorithm will choose it as the offloading target. Otherwise, the algorithm selects the vehicle with Eq.~(\ref{rrrrrr}) among the available vehicles, where $\mathcal{H}_{b_{g}}(t) = \sqrt{\frac{\mathscr{P}_g\left[T_{(i j) k}-\tau_{k}\right]^{+} \ln n_{(i j)}(t)}{\mathcal{J}_{(i j)}(t)}}$ is the width of the confidence interval of vehicle $b_{ig}$ at the time $t$. Hereafter, we verify $J^0_{exist}$ whether it is the empty set. If it is not, the transportation management server will reassign the bandwidth to recoup the deficient resource in $J^0_{exist}$.

\begin{algorithm} 
	\caption{SMTO Resource Allocation} \label{Ag2}
	$\bf{input:}$ $\bm{R}$, $Q_g(0)$, $J^0$, and $J^1$  \\
	\For{$j:$ road segment $1$ to $M$} 
	{
	\While{True}
	{

		\If{$J^0_j(g,b_{ig}) = \emptyset$}
		{
				Update reward $\displaystyle Q_{g}(t+1)$;

		\While{$b_{ig}$ has the parent}
		{
    		Update $Q_{\text{parent of } b_{ig}}(t+1) \leftarrow Q_{g}(t+1)$;
		}
		
		}
		
		\If{$\mathscr{S}(g,b_{ig})$ is a leaf node}
		{

			Select a child for application $g+1$; \\ $b_{i(g+1)} \leftarrow \underset{j \in \mathcal{N}_{i}}{\arg \max } \ Q_g(j) +\mathcal{H}_{b_{g}}(t)$;

		}

    $n_{b_{ig}}(t+1) \leftarrow n_{b_{ig}}(t) + 1$, for all $b_i(t)\in N_i(t)$
		
			\If{vehicle $j$ accepts the application $g+1$}
			{
				Append the child $j$ to the tree;\\
			    Update $\mathcal{J}_{(i j)}(t+1) \leftarrow \mathcal{J}_{(i j)}(t) + 1$;
			}

		\If{$J^0_{exist}\neq \emptyset$}
			{
				\If{$j \in J^0_{exist}$}
				{
					$R_j \leftarrow R_j + \Delta R^{exist}_j$;
				}
				
				\If{$j \in J^0_{empty}$}
				{
					$R_j \leftarrow R_j - \Delta R^{empty}_j$;
				}		
					
			}
			\Else{
          	   Break;		
		     }
			
	}
}
$\bm{output:}$ $\bm{R}$		
	
\end{algorithm}

%
%
%
%
%
%
%
%
%
%

\begin{figure*}
\centering
\subfloat[Impact of the safety distance on vehicle string stability.]{\label{adsfjlh_3}{\includegraphics[width=0.25\linewidth]{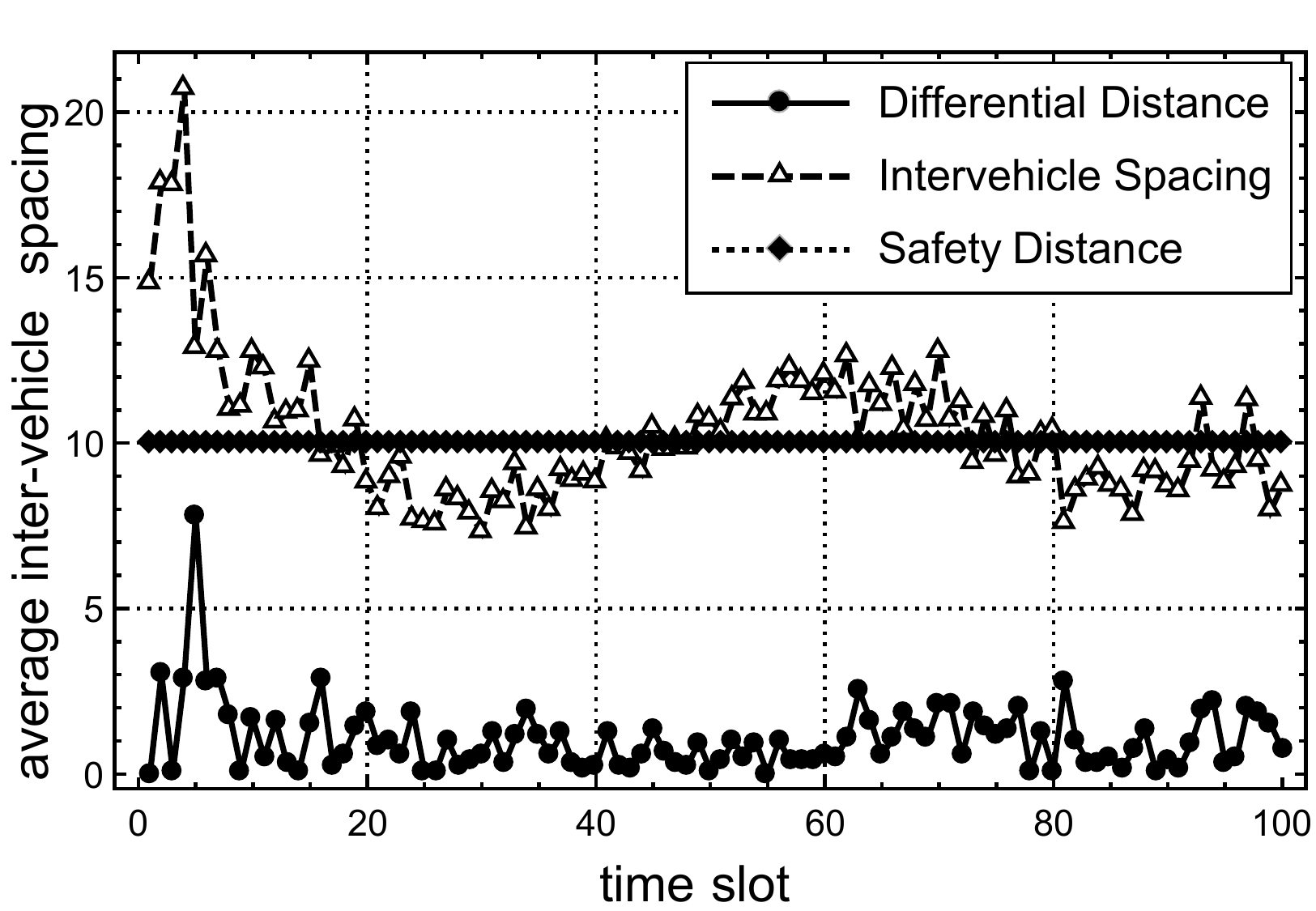}}} \hfill
\subfloat[Impact of the safety distance on throughput.]{\label{sfdvhjfdkjgkbjekwbu2i3}{\includegraphics[width=0.22\linewidth]{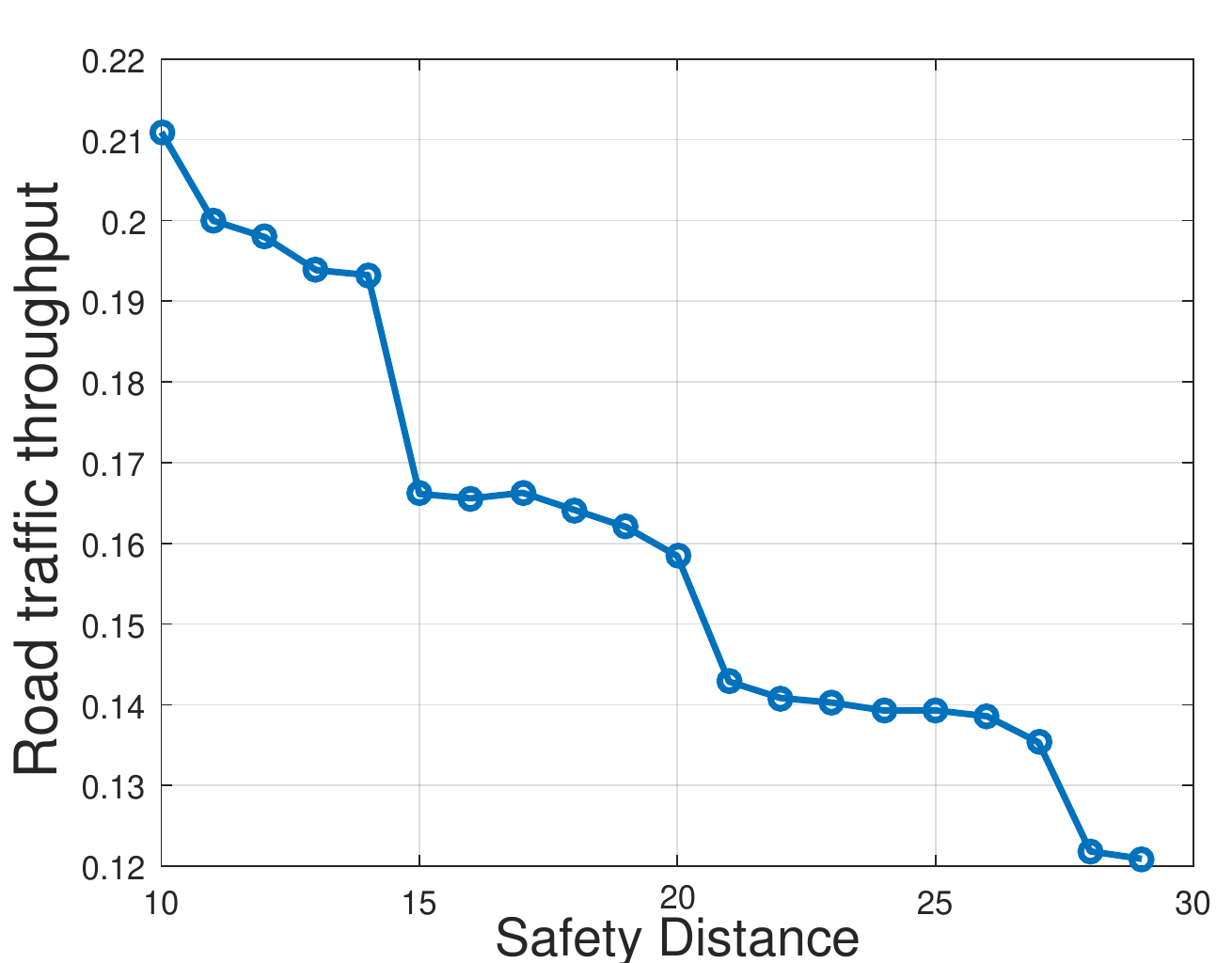}}} \hfill
\subfloat[Relation between string stability and road traffic throughput.]{\label{kjkkkkkk}{\includegraphics[width=0.25\linewidth]{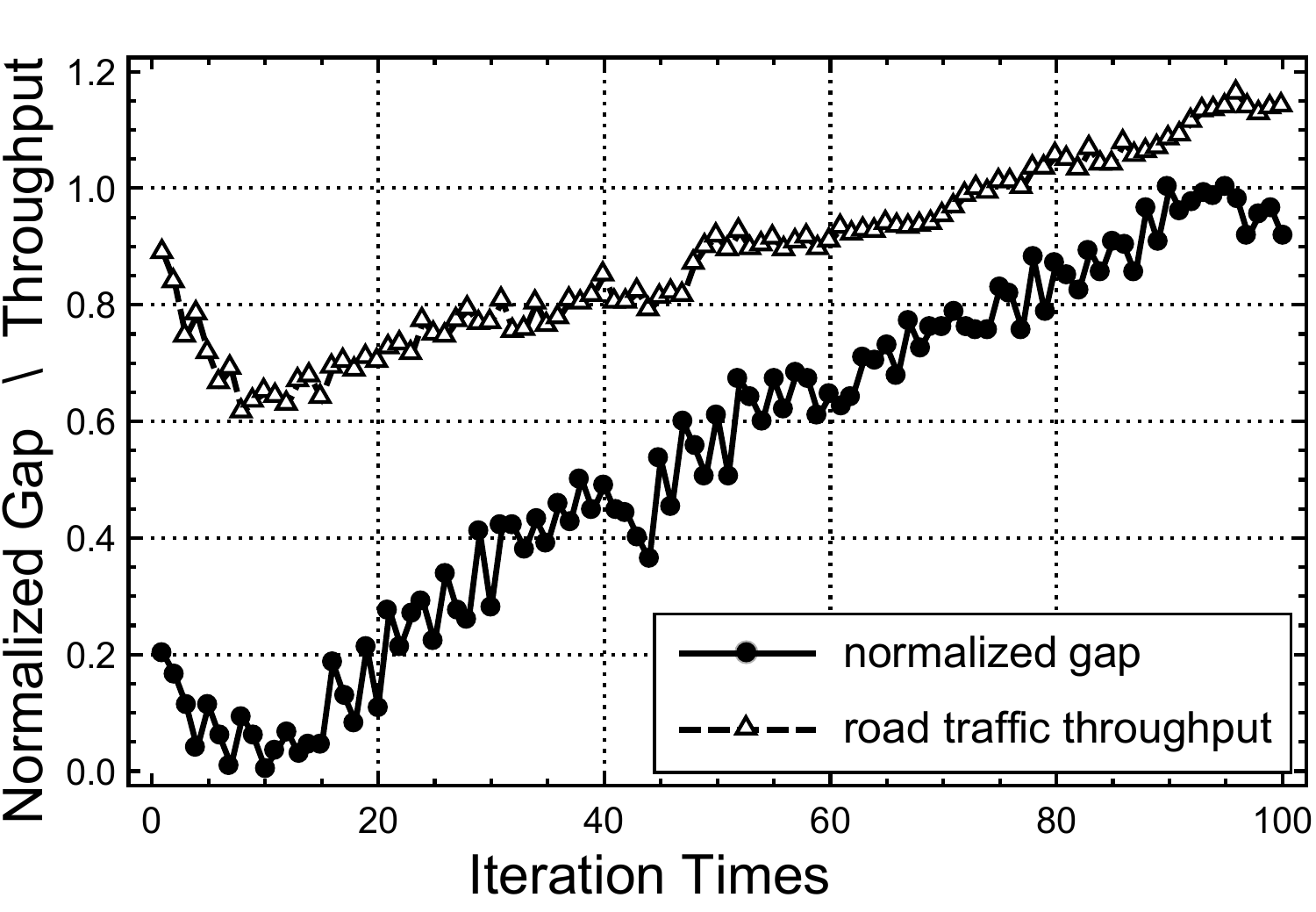}}}\hfill
\subfloat[Safety distance with different $\delta$ that represents the weight of string stability in $P1$.]{\label{adsljgabrkjedsfadfbgpvixpu}{\includegraphics[width=0.26\linewidth]{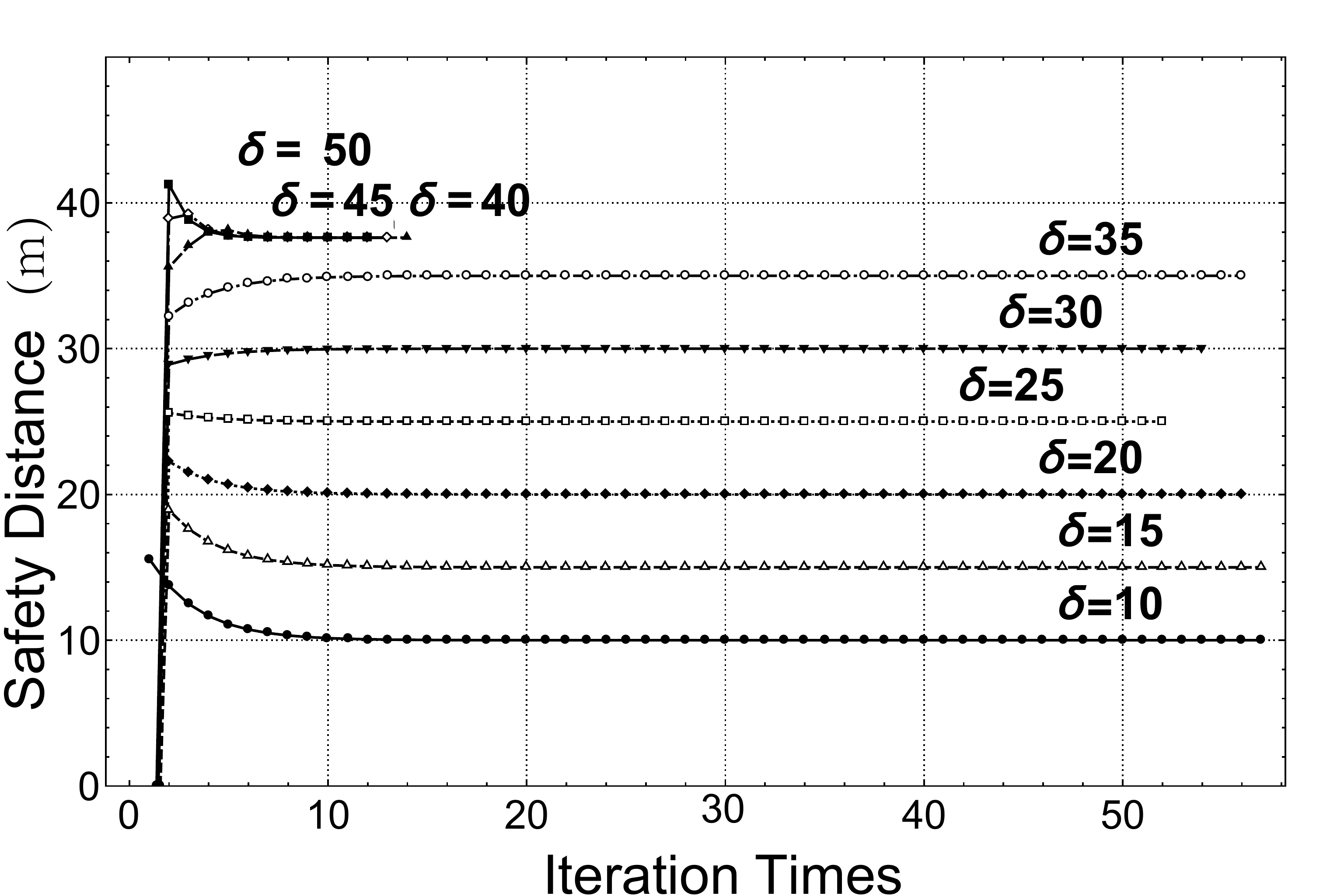}}} \hfill
\caption{Relations of the safety distance with string stability and road traffic throughput.}
\label{Per_of_Road_Traffic}
\end{figure*}

\section{Performance Evaluation}

To confirm the impact of safety distance on the string stability and road traffic throughput, we account for a Cellular Automata-based three-lane highway scenario, which has been investigated in a plethora of road traffic studies. Daoudia \textit{et al.} \cite{Daoudia2003Numerical} proposed a Cellular Automata-based three-lane version which takes into account of the exchange vehicles between the different lanes. However, this model did not account for the acceleration of vehicles. Li \textit{et al.} \cite{Li2016ACS} considered the heterogeneity of vehicle acceleration by the Cellular Automata, but the simulation is only suitable for the freeway traffic flow. Zamith \textit{et al.} \cite{Zamith2015A} defined the actions of a specific vehicle in the Cellular Automata traffic context that depicts the vehicle behaviors by the stochastic rules. However, vehicle behaviors are not always stochastic. Some of which should follow the traffic rules. In this section, we assume that a vehicle involves velocity updating (deceleration/acceleration), lane changing, and road congestion. These behaviors are elaborated below.


\subsubsection{Acceleration} if $v< v_{max} \; pixel/s$, ($v \ge 1 \; pixel/s$), and the distance with the preceding vehicle is lager than the safety distance, (the distance with the preceding vehicle is less than the safety distance), $v(t+1)=v(t)+1$, ($v(t+1)=v(t)-1$), where $v_{max} = 30$ $pixel/s$ is the limit speed in a particular road segment.
\subsubsection{Uniform speed} if $v = v_{max} \; pixel/s$ or the inter-vehicle spacing is equal to the safety distance, $v(t+1)=v(t)$.
\subsubsection{Lane Changing} If one adjacent lane has enough consecutive space, while there is heavy traffic on the current lane, the vehicle will go into the adjacent lane with a certain probability.

\subsubsection{Road congestion} road congestion is triggered when two successive vehicles touch each other. In this case, the velocity of the touched vehicles set to $0 \; pixel/s$.


\begin{figure} 
     \centering
     \includegraphics[width=0.48\textwidth]{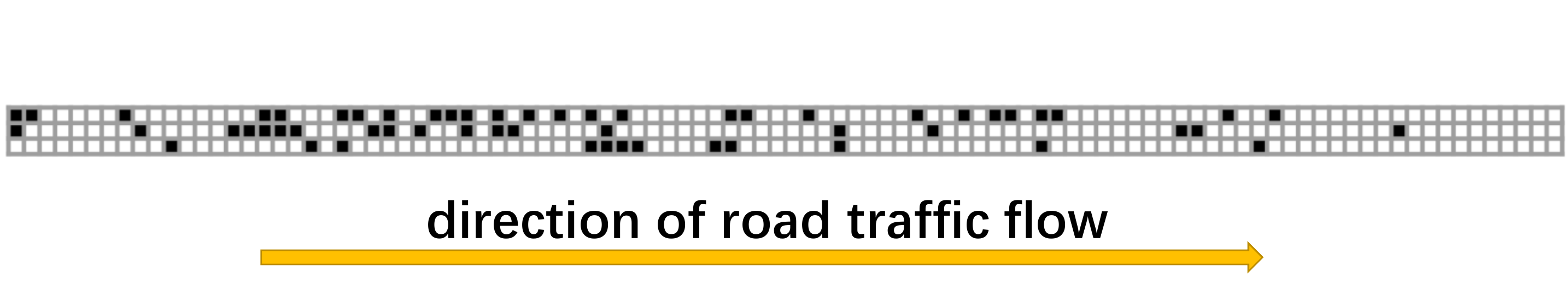} 
     \caption{Cellular Automation for road traffic. The length of the road segment is $100 \ pixels$, and the limit speed is $30 \ pixel/s$.} 
     \label{sksksks}
\end{figure}



%
%
%
%
%
%
%


The above-mentioned behaviors can precisely imitate the real vehicle behaviors on road. The simulation scenario is shown in Fig.~\ref{sksksks}, in which each black pixel represents a vehicle and each white pixel represents a unit empty space on road. The initial speed of the entry vehicles is set to $5$ $pixel/s$. The average arrival rate is $0.5$. The length of the road segment is $100$ $pixels$. Leveraging the Cellular Automata simulation, we investigate the impact of the safety distance on string stability and road traffic throughput in Fig.~\ref{adsfjlh_3} and Fig.~\ref{sfdvhjfdkjgkbjekwbu2i3}, respectively.

\subsection{The impact of safety distance on comfort and throughput}
In Fig.~\ref{adsfjlh_3}, the differential distance $Dd(t)$ represents the difference of the average inter-vehicle spacing in two successive time slots, i.e., $Dd(t) = \|\overline{s}(t)-\overline{s}(t-1)\|$, where $\overline{s}(t)$ is the average inter-vehicle spacing in a road segment at time $t$. In general, the differential distance reflects the fluctuation of the inter-vehicle spacing that represents the instability of vehicle string. The large differential distance results in a heavy fluctuation of vehicle string and deteriorating ride comfort and energy efficiency. When the average inter-vehicle spacing $\overline{s}(t)$ closes to the safety distance $s^*$, as shown in Fig.~\ref{adsfjlh_3}, the differential distance becomes small in time interval $[20, 100]$. However, in time interval $[0, 20]$, the average inter-vehicle spacing is far away from the safety distance, which results in a large fluctuation of differential distance. Meanwhile, the vehicle string becomes unstable and deteriorates the ride comfort of passengers. Fig.~\ref{sfdvhjfdkjgkbjekwbu2i3} illustrates that road traffic throughput is inversely proportional to the safety distance. The reason is that the safety distance is the equilibrium inter-vehicle spacing. And, the road traffic throughput density is determined by the equilibrium inter-vehicle spacing and velocity. Moreover, larger equilibrium vehicle spacing derives the lower road traffic throughput. Therefore, the safety distance is inversely proportional to the road traffic throughput.


\subsection{The relation between string stability and throughput}
Hereafter, we investigate the relation between string stability and road traffic throughput. To carefully compare the string stability metric $ \|\frac{1}{\rho} - s^*\|$ with road traffic throughput, we introduce the normalized gap $d_s$, which is defined as

\begin{equation}
\begin{aligned}
\begin{split}
d_s =  \min \left\{\frac{\|\frac{1}{\rho} - s^*\|}{\max\{\|\frac{1}{\rho} - s^*\|,\omega\}}, 1\right\},
\end{split} 
\end{aligned} 
\label{dadsalhasbdlnvs}
\end{equation}

\noindent where $\omega$ is an infinitesimal number to avoid zero denominator. Road traffic throughput compared with the normalized gap is demonstrated in Fig.~\ref{kjkkkkkk}. Road traffic throughput increases with the normalized gap. The large normalized gap represents the heavy fluctuation of a vehicle string that deteriorates string stability. This result verifies our analysis conclusion: the transportation operator cannot optimize road traffic throughput without considering the impact on the vehicle string stability.

\begin{figure*}
\centering
\subfloat[Execution time.]{\label{dsfjhjklhalllokoko}{\includegraphics[width=0.25\linewidth]{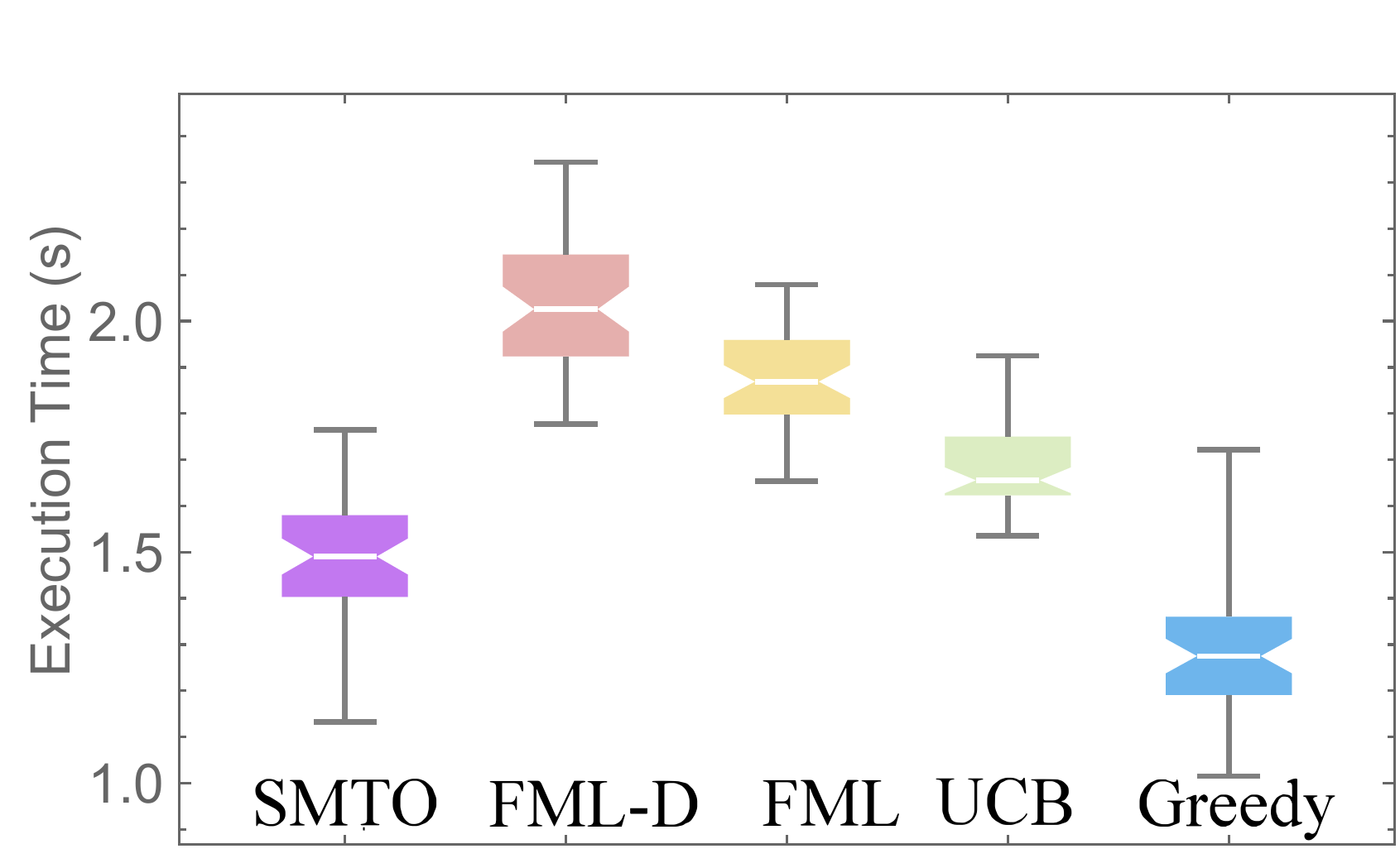}}} \hfill
\subfloat[Average offloading delay.]{\label{kjhkjhkjasappp}{\includegraphics[width=0.25\linewidth]{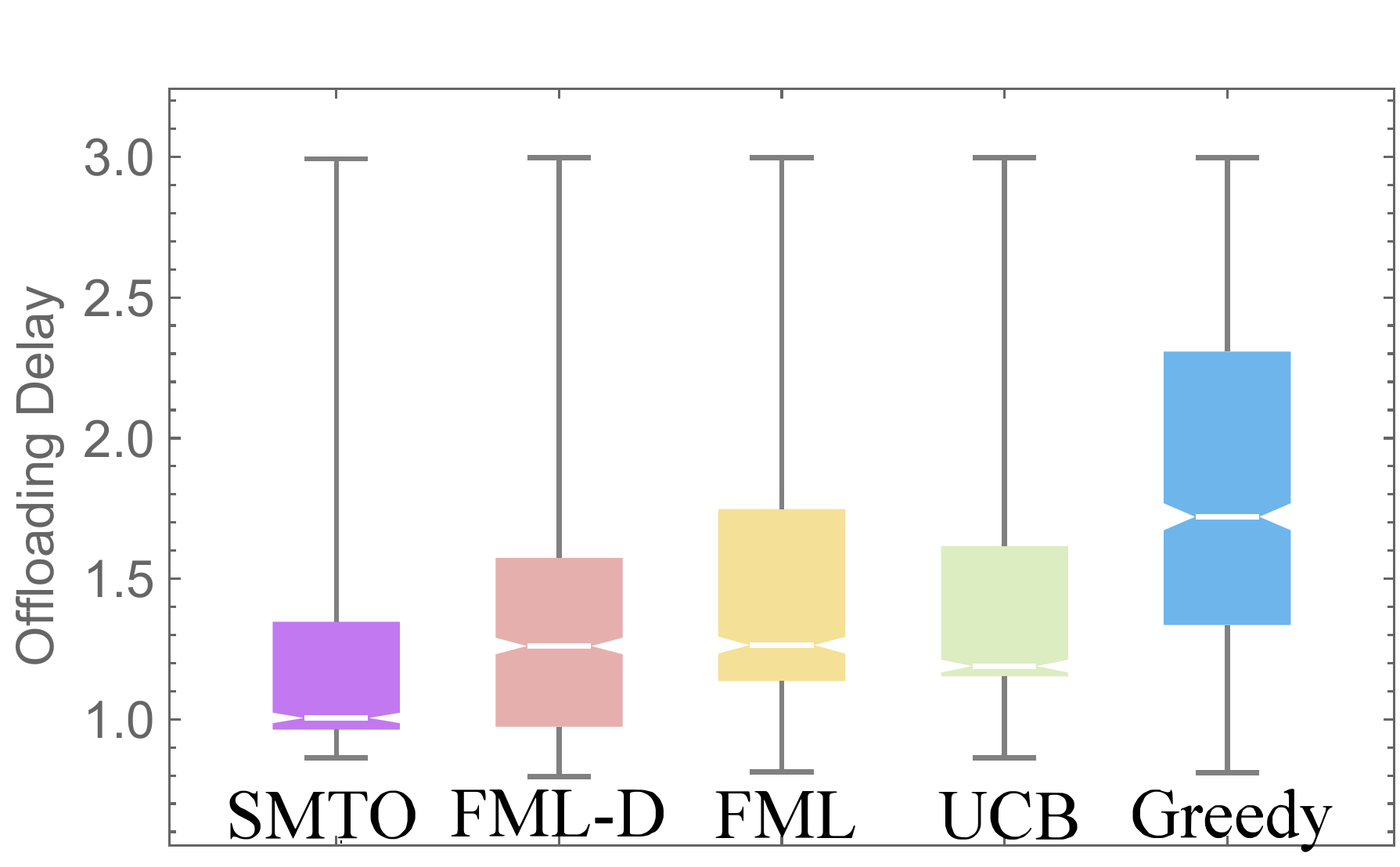}}} \hfill
\subfloat[Acceptance ratio distribution.]{\label{vvkvkvkvkvkvvsadada}{\includegraphics[width=0.25\linewidth]{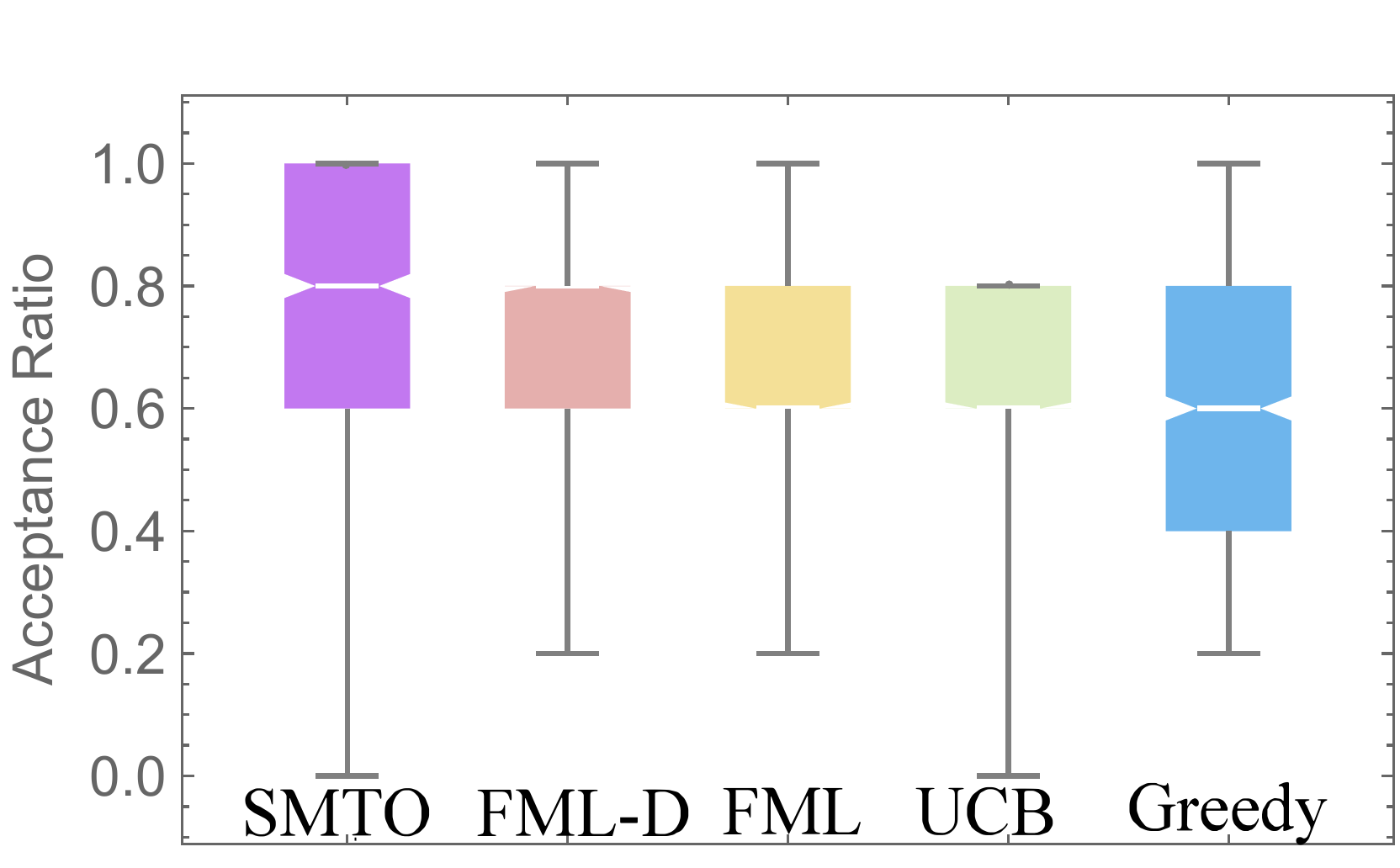}}}\hfill
\subfloat[Rewards distribution.]{\label{dasjhnnn}{\includegraphics[width=0.25\linewidth]{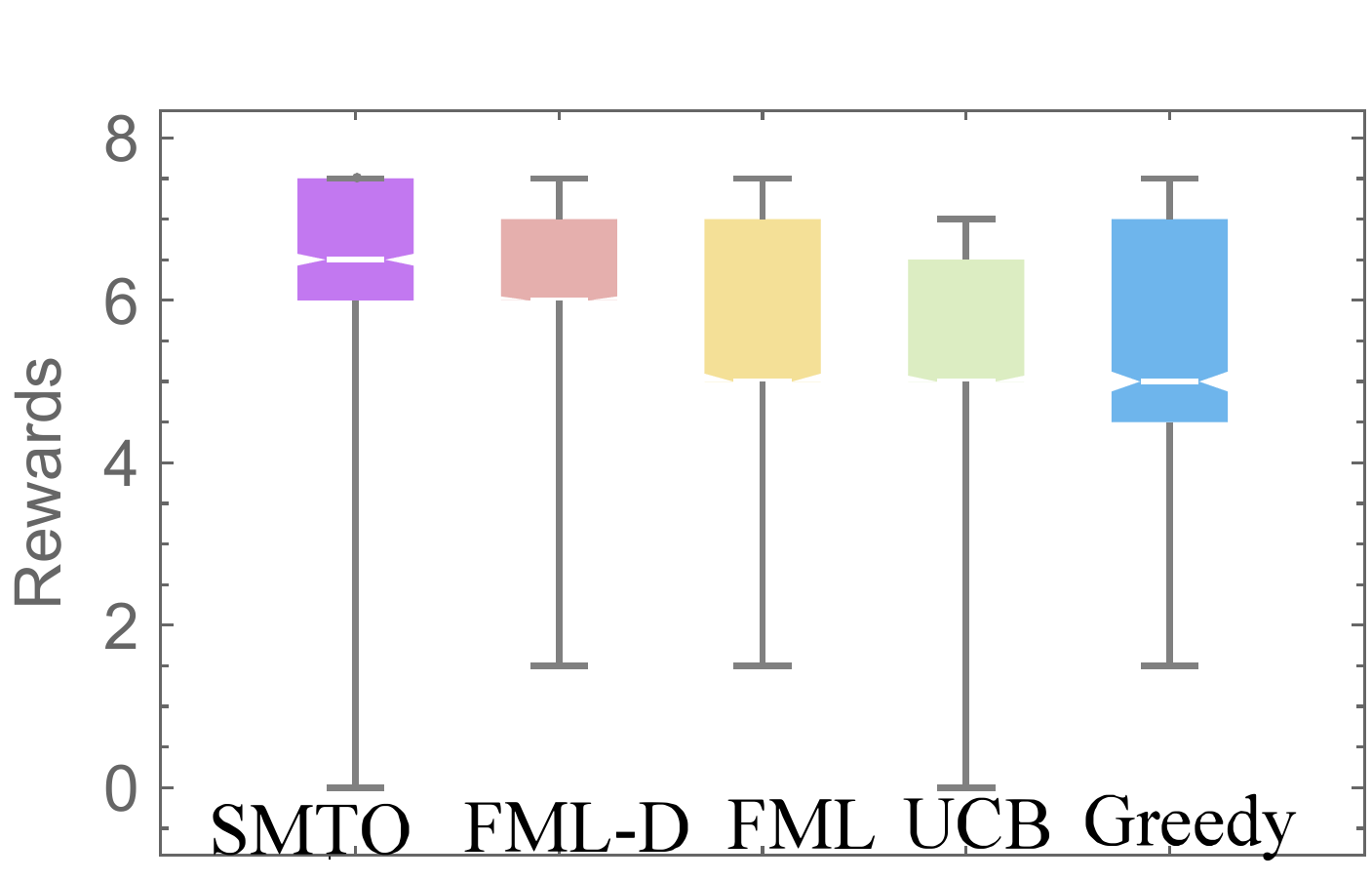}}} \hfill
\caption{Performances of different algorithms.}
\label{fig_sim_times}
\end{figure*}

The numerical simulation of Alg.~\ref{Ag1} is depicted in Fig.~\ref{adsljgabrkjedsfadfbgpvixpu}. Here, we investigate five road segments whose initial road densities are randomly generated from $[0.02,0.1] \ (m^{-1})$. There are 3 vehicular applications to offload. The arrival rate of each application is randomly generated from $[0.2, 1.5]$. $\frac{y_i}{\mu}$ and $z$ are both initialized to $1$. As shown in Fig.~\ref{adsljgabrkjedsfadfbgpvixpu}, each curve represents an asymptotic result of Alg.~{\ref{Ag1}} with a certain $\delta$. Fig.~\ref{adsljgabrkjedsfadfbgpvixpu} demonstrates the optimal safety distance increased with the value of $\delta$. This is because $\delta$ represents the weight of the stability of road traffic in the optimization. According to the optimization (\ref{gfg}), larger $\delta$ is prone to minimize $ \|z - \frac{1}{\rho_i}\|$, where $z = \frac{1}{M}\sum\limits_{i = 1}^{M}{\frac{1}{\rho_i}}$ as $\delta \rightarrow \infty$. In our simulation, $\frac{1}{M}\sum\limits_{i = 1}^{M}{\frac{1}{\rho_i}} = 37.61$. Thus, the optimal safety distance approaches to $37.61$ when $\delta$ becomes large. However, if $\delta$ is small, the algorithm prefers to minimize $s^*$ for road throughput. When $\delta$ is over $40$, the affect of $\frac{1}{2}\sum_{i}^{M} \left\|s^*_i \right\|_2^2 $ can be neglected in Eq.~(\ref{gfg}). The optimization model $P3$ reduces to a simple $\ell_1$-Norm problem, which results in a fast convergence of the curves $\delta = 40$, $\delta = 45$, and $\delta = 50$. Therefore, transportation management server can adapt the different values of $\delta$ to obtain the favourable vehicle string stability and road traffic throughput.



In addition, we investigate our proposed SMTO algorithm in terms of execution time, average offloading delay, acceptance ratio, and rewards, compared with that of the FML algorithm, FML-D algorithm, traditional UCB algorithm, and the Greedy algorithm. FML is proposed by Sim \textit{et al.} \cite{{8472783},{7775114}}, which is a contextual multi-armed bandit online learning algorithm. FML-D is the FML algorithm combined with the upper bound delay information. In the FML-D algorithm, the exploitation process is revised as $b_{i(g+1)} \leftarrow \underset{j \in \mathcal{N}_{i}}{\arg \max } \ \ Q_g(j) +\sqrt{[\tau_k-T_{(ij)k}]^+}$, where $Q_g(t)$ represents the returned reward of completing the application $g$ at time $t$. The UCB algorithm is a classical learning algorithm for multi-armed bandit problems \cite{Auer2002}. However, it does not take the account of upper bound delay $T_{(ij)k}$. Comparing with the SMTO, the Greedy algorithm does not consider the upper bound delay and always stays in the exploitation process.

Due to the limited communication bandwidth, the maximum number of vehicles of the platoon is set to $5$ in our simulation. The leaving probability of each vehicle in the platoon is followed by the experiential distribution, where the average leaving rate is $0.2$. The initial number of vehicles in the platoon is $3$. The computing resource (CPU frequency) of each vehicle is randomly selected from $[2, 10]$ $Mbps$. The sharing communication bandwidth of the platoon is $10$ $Mbps$. The category of applications is set to $5$. The data volume of each application is generated from $[1,5]$ $(Mb)$, uniformly. Each application has own delay requirement that is selected from the range of $[1,3]$ seconds. Moreover, each application has own priority that represents the importance of this category application. When a high priority application has been finished, it feedbacks a high reward. The summation of the feedback rewards is denoted by $Q_g$ in Eq.~(\ref{rrrrrr}). In the simulation, each algorithm implements 1000 times to obtain stable statistical results. The simulation is implemented by Wolfram Mathematica on a laptop with i5-8300h CPU and 16G RAM.

\subsection{Execution Time}

The execution time of different algorithms is illustrated in Fig.~\ref{dsfjhjklhalllokoko}. The execution time is collected from 1000 times offloading process since the SMTO, FML-D, FML, and UCB algorithms need enough simulation time to train. For fairly comparing, the Greedy algorithm also implements 1000 times. The Greedy algorithm has the lowest execution time because of the simplicity. The execution time of the proposed SMTO is lower than the UCB, FML, and FML-D algorithms. Since the collection of the upper bound delay information requires extra time, the FML-D has the highest execution time for decision making. Because of the extra upper bound delay information and simple search structure, SMTO is faster than the UCB, FML, and FML-D upon the exploitation process that results in the low execution time of SMTO.

\subsection{Offloading Delay}

Fig.~\ref{kjhkjhkjasappp} depicts the offloading delay of different algorithms. In our simulation, the offloading delay is composed of the transmission delay and processing delay \cite{Kato8322166}. If an offloading delay excesses the requirement delay of the application, we regard the double times of the requirement delay as the offloading delay. The SMTO algorithm achieves the least offloading delay. Comparing to the other algorithms, the offloading delay of the Greedy algorithm is unstable.

\subsection{Acceptance Ratio}

Next, we investigate the impact of the different algorithms on the acceptance ratio and rewards of the V2V related applications. Fig.~\ref{vvkvkvkvkvkvvsadada} illustrates the box-plot of applications acceptance ratio for different algorithms. The offloading acceptance ratio (AR) is given as
\begin{equation}
AR = \frac{ N^{\mathrm{accept}}}{N^{\mathrm{total}}},\label{e20_1} 
\end{equation}

\noindent where $N^{\mathrm{total}} $ and $N^{\mathrm{accept}} $ represent the number of the total arrived applications and the number of accepted applications, respectively. The y-axis of Fig.~\ref{vvkvkvkvkvkvvsadada} is the acceptance ratio. The expectation of acceptance ratio of SMTO is the largest, and that of Greedy algorithm is the smallest. In addition, acceptance ratio of FML-D, FML, and UCB are very similar.


\subsection{Rewards}

The box-plot of the rewards versus different algorithms is illustrated in Fig.~\ref{dasjhnnn}, where the y-axis is the average rewards. Each category application has its own reward to courage vehicles to complete this category application in time. In general, the rewards of the CACC applications are higher than that of the lane change assist applications. In our simulation, there are $5$ categories applications whose the rewards are assumed as $\{0.5, 1, 1.5, 2, 2.5\}$, respectively. In Fig.~\ref{dasjhnnn}, the average reward of SMTO is the largest. The Greedy algorithm has better rewards than that of the UCB. However, the fluctuation of Greedy algorithm is the largest compared with that of the other four algorithms. The reason is that the Greedy algorithm only invokes a simple maximum rewards exploitation, which is not suitable for the high dynamic vehicular environment.


\section{Conclusion}
In this paper, we analyze the properties of driving safety, vehicle string stability, and road traffic throughput, as well as the relationship between them. Then, the joint optimization of these road metrics is formulated in terms of resource management. It can be found that the vehicle string stability and road traffic throughput are coupled with each other upon the precondition of safe driving. Optimizing one of the road traffic metrics cannot avoid the influence on the other one. In addition, communication bandwidth and on-board computing can be regarded as the control variables for these road traffic metrics from the perspective of resource management. With a given amount of communication and computing resources, the upper bound delay of V2V offloading can be determined based on the NC theory. The obtained upper bound delay is helpful for the transportation planner to improve road traffic performances. As a future work, we will study a comprehensive scheme taking an account of the Cellular-based Vehicle-to-Everything (C-V2X) communication in transportation management.


\appendices

\footnotesize
\bibliography{biblio}

\begin{thebibliography}{10}
\providecommand{\url}[1]{#1}
\csname url@samestyle\endcsname
\providecommand{\newblock}{\relax}
\providecommand{\bibinfo}[2]{#2}
\providecommand{\BIBentrySTDinterwordspacing}{\spaceskip=0pt\relax}
\providecommand{\BIBentryALTinterwordstretchfactor}{4}
\providecommand{\BIBentryALTinterwordspacing}{\spaceskip=\fontdimen2\font plus
\BIBentryALTinterwordstretchfactor\fontdimen3\font minus
  \fontdimen4\font\relax}
\providecommand{\BIBforeignlanguage}[2]{{%
\expandafter\ifx\csname l@#1\endcsname\relax
\typeout{** WARNING: IEEEtran.bst: No hyphenation pattern has been}%
\typeout{** loaded for the language `#1'. Using the pattern for}%
\typeout{** the default language instead.}%
\else
\language=\csname l@#1\endcsname
\fi
#2}}
\providecommand{\BIBdecl}{\relax}
\BIBdecl

\bibitem{8320295}
X.~{Liang}, T.~{Yan} \emph{et~al.}, ``A distributed intersection management
  protocol for safety, efficiency, and driver’s comfort,'' \emph{IEEE
  Internet of Things Journal}, vol.~5, no.~3, pp. 1924--1935, 2018.

\bibitem{Kato8584062}
J.~{Wang}, J.~{Liu}, and N.~{Kato}, ``{Networking and Communications in
  Autonomous Driving: A Survey},'' \emph{{IEEE Communications Surveys
  Tutorials}}, vol.~21, no.~2, pp. 1243--1274, Secondquarter 2019.

\bibitem{liulin2018}
C.~{Liu}, C.~w.~{Lin} \emph{et~al.}, ``Improving efficiency of autonomous
  vehicles by v2v communication,'' in \emph{2018 Annual American Control
  Conference (ACC)}, 2018.

\bibitem{Dunbar5876300}
W.~B. {Dunbar}, D.~S. {Caveney} \emph{et~al.}, ``Distributed receding horizon
  control of vehicle platoons: Stability and string stability,'' \emph{IEEE
  Transactions on Automatic Control}, pp. 620--633, 2012.

\bibitem{7349170}
P.~{Nilsson}, O.~{Hussien} \emph{et~al.}, ``Correct-by-construction adaptive
  cruise control: Two approaches,'' \emph{IEEE Transactions on Control Systems
  Technology}, vol.~24, no.~4, pp. 1294--1307, 2016.

\bibitem{8798668}
Y.~{Liu}, H.~{Yu} \emph{et~al.}, ``{Deep Reinforcement Learning for Offloading
  and Resource Allocation in Vehicle Edge Computing and Networks},'' \emph{IEEE
  Transactions on Vehicular Technology}, pp. 1--1, 2019.

\bibitem{Qiao044}
G.~{Qiao}, S.~{Leng} \emph{et~al.}, ``Collaborative task offloading in
  vehicular edge multi-access networks,'' \emph{IEEE Communications Magazine},
  vol.~56, no.~8, pp. 48--54, August 2018.

\bibitem{Gallego8845107}
F.~{Vázquez-Gallego}, R.~{Vilalta} \emph{et~al.}, ``Demo: A mobile edge
  computing-based collision avoidance system for future vehicular networks,''
  in \emph{IEEE INFOCOM 2019 - IEEE Conference on Computer Communications
  Workshops (INFOCOM WKSHPS)}, 2019, pp. 904--905.

\bibitem{8317758}
E.~{van Nunen}, J.~{Verhaegh} \emph{et~al.}, ``{Robust model predictive
  cooperative adaptive cruise control subject to V2V impairments},'' in
  \emph{IEEE 20th International Conference on Intelligent Transportation
  Systems (ITSC)}, 2017, pp. 1--8.

\bibitem{Nekoui2010Fundamental}
Nekoui \emph{et~al.}, ``Fundamental tradeoffs in vehicular ad hoc networks,''
  in \emph{ACM International Workshop on Vehicular Internetworking}, 2010.

\bibitem{Qiao573}
G.~{Qiao}, S.~{Leng} \emph{et~al.}, ``Deep reinforcement learning for
  cooperative content caching in vehicular edge computing and networks,''
  \emph{IEEE Internet of Things Journal}, 2019.

\bibitem{Kato8847416}
T.~K. {Rodrigues}, K.~{Suto} \emph{et~al.}, ``{Machine Learning meets
  Computation and Communication Control in Evolving Edge and Cloud: Challenges
  and Future Perspective},'' \emph{accepted by IEEE Communications Surveys
  Tutorials}, 2019.

\bibitem{7277110}
K.~{Katsaros}, M.~{Dianati} \emph{et~al.}, ``End-to-end delay bound analysis
  for location-based routing in hybrid vehicular networks,'' \emph{IEEE
  Transactions on Vehicular Technology}, vol.~65, no.~9, pp. 7462--7475, 2016.

\bibitem{Shao119}
C.~{Shao}, S.~{Leng} \emph{et~al.}, ``Performance analysis of connectivity
  probability and connectivity-aware mac protocol design for platoon-based
  vanets,'' \emph{IEEE Transactions on Vehicular Technology}, vol.~64, no.~12,
  pp. 5596--5609, Dec 2015.

\bibitem{Wang8360847}
R.~{Wang}, J.~{Yan}, D.~{Wu}, H.~{Wang}, and Q.~{Yang}, ``Knowledge-centric
  edge computing based on virtualized d2d communication systems,'' \emph{IEEE
  Communications Magazine}, pp. 32--38, 2018.

\bibitem{Kato8361406}
H.~{Guo}, J.~{Liu} \emph{et~al.}, ``Mobile-edge computation offloading for
  ultradense iot networks,'' \emph{IEEE Internet of Things Journal}, pp.
  4977--4988, 2018.

\bibitem{Lian7322289}
Y.~{Lian}, Y.~{Zhao} \emph{et~al.}, ``Longitudinal collision avoidance control
  of electric vehicles based on a new safety distance model and
  constrained-regenerative-braking-strength-continuity braking force
  distribution strategy,'' \emph{IEEE Transactions on Vehicular Technology},
  pp. 4079--4094, 2016.

\bibitem{Tosin2009}
B.~Piccoli, A.~Tosin \emph{et~al.}, \emph{Vehicular Traffic: A Review of
  Continuum Mathematical Models}, 01 2009, pp. 9727--9749.

\bibitem{8644035}
K.~{Xiong}, S.~{Leng} \emph{et~al.}, ``{Smart Network Slicing for Vehicular
  Fog-RANs},'' \emph{IEEE Transactions on Vehicular Technology}, vol.~68,
  no.~4, pp. 3075--3085, 2019.

\bibitem{BoydS2010}
S.~Boyd, N.~Parikh \emph{et~al.}, ``Distributed optimization and statistical
  learning via the alternating direction method of multipliers,''
  \emph{Foundations \& Trends in Machine Learning}, pp. 1--122, 2010.

\bibitem{Zhou8667693}
Z.~{Zhou}, J.~{Feng} \emph{et~al.}, ``Energy-efficient edge computing service
  provisioning for vehicular networks: A consensus admm approach,'' \emph{IEEE
  Transactions on Vehicular Technology}, pp. 5087--5099, 2019.

\bibitem{5967982}
E.~{Felemban}, E.~{Ekici} \emph{et~al.}, ``{Single Hop IEEE 802.11 DCF Analysis
  Revisited: Accurate Modeling of Channel Access Delay and Throughput for
  Saturated and Unsaturated Traffic Cases},'' \emph{IEEE Transactions on
  Wireless Communications}, vol.~10, no.~10, pp. 3256--3266, 2011.

\bibitem{Kato7636965}
T.~G. {Rodrigues}, K.~{Suto} \emph{et~al.}, ``{Hybrid Method for Minimizing
  Service Delay in Edge Cloud Computing Through VM Migration and Transmission
  Power Control},'' \emph{IEEE Transactions on Computers}, vol.~66, no.~5, pp.
  810--819, May 2017.

\bibitem{Jiang2008Stochastic}
Y.~Jiang, Y.~Liu \emph{et~al.}, ``Stochastic network calculus,'' 2008.

\bibitem{Katsaros2016End}
K.~Katsaros and M.~Dianati, ``End-to-end delay bound analysis for
  location-based routing in hybrid vehicular networks,'' \emph{IEEE
  Transactions on Vehicular Technology}, vol.~65, pp. 7462--7475, 2016.

\bibitem{Wang7999188}
R.~{Wang}, K.~{Liu}, D.~{Wu}, H.~{Wang}, and J.~{Yan},
  ``Malicious-behavior-aware d2d link selection mechanism,'' \emph{IEEE
  Access}, pp. 15\,162--15\,173, 2017.

\bibitem{8080373}
R.~{Molina-Masegosa}, J.~{Gozalvez} \emph{et~al.}, ``{LTE-V for Sidelink 5G V2X
  Vehicular Communications: A New 5G Technology for Short-Range
  Vehicle-to-Everything Communications},'' \emph{IEEE Vehicular Technology
  Magazine}, vol.~12, no.~4, pp. 30--39, 2017.

\bibitem{Kleinberg2010}
R.~Kleinberg and A.~Niculescu-Mizil, ``Regret bounds for sleeping experts and
  bandits,'' \emph{Machine Learning}, pp. 245--272, 2010.

\bibitem{Kato8657791}
J.~{Wang}, L.~{Zhao} \emph{et~al.}, ``{Smart Resource Allocation for Mobile
  Edge Computing: A Deep Reinforcement Learning Approach},'' \emph{{IEEE
  Transactions on Emerging Topics in Computing}}, 2019.

\bibitem{Daoudia2003Numerical}
A.~K. Daoudia, N.~Moussa \emph{et~al.}, ``Numerical simulations of a three-lane
  traffic model using cellular automata,'' \emph{Chinese Journal of Physics},
  vol.~41, no.~6, pp. 671--681, 2003.

\bibitem{Li2016ACS}
Q.~L. Li, S.~C. Wong \emph{et~al.}, ``A cellular automata traffic flow model
  considering the heterogeneity of acceleration and delay probability,''
  \emph{Physica A Statistical Mechanics \& Its Applications}, pp. 128--134,
  2016.

\bibitem{Zamith2015A}
M.~Zamith, Leal-Toledo \emph{et~al.}, ``A new stochastic cellular automata
  model for traffic flow simulation with drivers’ behavior prediction,''
  \emph{Journal of Computational Science}, vol.~9, no. July 2015, pp. 51--56,
  2015.

\bibitem{8472783}
G.~H. {Sim}, S.~{Klos} \emph{et~al.}, ``{An Online Context-Aware Machine
  Learning Algorithm for 5G mmWave Vehicular Communications},'' \emph{IEEE/ACM
  Transactions on Networking}, vol.~26, no.~6, pp. 2487--2500, Dec 2018.

\bibitem{7775114}
S.~{Müller}, O.~{Atan} \emph{et~al.}, ``Context-aware proactive content
  caching with service differentiation in wireless networks,'' \emph{IEEE
  Transactions on Wireless Communications}, vol.~16, no.~2, pp. 1024--1036, Feb
  2017.

\bibitem{Auer2002}
P.~Auer, N.~Cesa-Bianchi \emph{et~al.}, ``Finite-time analysis of the
  multiarmed bandit problem,'' \emph{Machine Learning}, pp. 235--256, May 2002.

\bibitem{Kato8322166}
T.~G. {Rodrigues}, K.~{Suto} \emph{et~al.}, ``{Cloudlets Activation Scheme for
  Scalable Mobile Edge Computing with Transmission Power Control and Virtual
  Machine Migration},'' \emph{IEEE Transactions on Computers}, vol.~67, no.~9,
  pp. 1287--1300, Sep 2018.

\end{thebibliography}

\end{document}